\documentclass[11pt]{article}
\usepackage[utf8]{inputenc}
\usepackage[american]{babel}
\usepackage[T1]{fontenc}
\usepackage{amssymb}
\usepackage{amsmath}
\usepackage{listings}
\usepackage{amsthm}
\usepackage{graphicx}
\usepackage{color}
\usepackage{xspace}
\usepackage{todonotes}
\usepackage{enumerate}
\usepackage{xargs}
\usepackage{tabularx}
\usepackage{multirow}
\usepackage{listings}
\usepackage{stmaryrd}
\usepackage{mathtools}  
\usepackage{float}

\usepackage[plainpages=false]{hyperref}
\hypersetup{
	unicode=false,
	colorlinks=true,
	linkcolor=blue,
	citecolor=blue,
	filecolor=blue,
	urlcolor=blue,
	pdfauthor={Andreas Cord-Landwehr, Matthias Fischer, Daniel Jung, Friedhelm Meyer auf der Heide},
	pdfsubject={Asymptotically Optimal Gathering on a Grid},
	pdftitle={Asymptotically Optimal Gathering on a Grid},
	pdfkeywords={gathering problem, autonomous robots, distributed algorithms, local algorithms, mobile agents, runtime bound, swarm formation problems}
}

\newtheorem{theorem}{Theorem}
   \newtheorem{lemma}{Lemma}
   
\newtheorem{definition}{Definition}

\newcommandx*{\LDAUOmicron}[2][1=@pkling_false]{\mathcal{O}\ifthenelse{\equal{#1}{small}}{\bigl(#2\bigr)}{\left(#2\right)}}
\newcommandx*{\LDAUomicron}[2][1=@pkling_false]{\mathrm{o}\ifthenelse{\equal{#1}{small}}{\bigl(#2\bigr)}{\left(#2\right)}}
\newcommandx*{\LDAUOmega}[2][1=@pkling_false]{\Omega\ifthenelse{\equal{#1}{small}}{\bigl(#2\bigr)}{\left(#2\right)}}
\newcommandx*{\LDAUomega}[2][1=@pkling_false]{\omega\ifthenelse{\equal{#1}{small}}{\bigl(#2\bigr)}{\left(#2\right)}}
\newcommandx*{\LDAUTheta}[2][1=@pkling_false]{\Theta\ifthenelse{\equal{#1}{small}}{\bigl(#2\bigr)}{\left(#2\right)}}

\newcommandx*{\set}[2][2=@pkling_false]{\left\{#1\ifthenelse{\equal{#2}{@pkling_false}}{}{\;\middle|\;#2}\right\}}

\newcommand{\oBdA}{w.l.o.g.\xspace}

\newcommand{\calO}{\mathcal{O}}

\newcommand{\viewingradius}{20}
\newcommand{\viewingradiusident}{11}
\newcommand{\pipelininginterval}{22}
\newcommand{\crossdistance}{3}
\newcommand{\figscale}{0.9}
\newcommand{\figscalee}{0.9}

\makeatletter
\xspaceaddexceptions{\check@icr}
\makeatother

\title{Asymptotically Optimal Gathering on a Grid}
\author{
    Andreas Cord-Landwehr
    \and
    Matthias Fischer
    \and
    Daniel Jung
    \and
    Friedhelm Meyer auf der Heide
}

\date{
    Heinz Nixdorf Institute \& Computer Science Department\\[0.2em]
    University of Paderborn (Germany)\\[0.2em]
    Fürstenallee 11, 33102 Paderborn\\[0.2em]
\quad\\[0.2em]
    \texttt{\{cola,mafi,daniel.jung,fmadh\}@uni-paderborn.de}
}

\begin{document}
\bibliographystyle{alphadin}
\maketitle
\thispagestyle{empty}
\begin{abstract}
    In this paper, we solve the local gathering problem of a swarm of $n$ indistinguishable, point-shaped
    robots on a two dimensional grid in asymptotically optimal time $\calO(n)$ in the fully synchronous $\mathcal{FSYNC}$ time model.
    Given an arbitrarily distributed (yet connected) swarm of robots, the gathering problem on the grid is to locate all robots within a $2\times 2$-sized area that is not known beforehand.
    Two robots are connected if they are vertical or horizontal neighbors on the grid.
    The locality constraint means that no global control, no compass, no global communication and only local vision is available; hence, a robot can only see its grid neighbors up to a constant $L_1$-distance, which also limits its movements.
    A robot can move to one of its eight neighboring grid cells and if two or more robots move to the same location they are \emph{merged} to be only one robot.
    The locality constraint is the significant challenging issue here, since robot movements must not harm the (only globally checkable) swarm connectivity.
    For solving the gathering problem, we provide a synchronous algorithm -- executed by every robot -- which ensures that robots merge without breaking the swarm connectivity.
    In our model, robots can obtain a special state, which marks such a robot to be performing specific connectivity preserving movements in order to allow later merge operations of the swarm.
    Compared to the grid, for gathering in the Euclidean plane for the same robot and time model the best known upper bound is $\calO(n^2)$ \cite{gatheringthetanquadrat}.
\end{abstract}
\textbf{Keywords: gathering problem, autonomous robots, distributed algorithms, local algorithms, mobile agents, runtime bound, swarm formation problems}
\section{Introduction}
    In the field of robot formation problems, there is a strong interest in what is feasible when robots only have very limited capabilities.
    Runtime bounds are intrinsically hard for such problems and tight results are rarely provided.
    One common formation problem is the gathering, where robots have to gather in one, beforehand unspecified point.

    There are essentially three significant approaches that try to tackle this problem:
    The first approach \cite{gatheringthetanquadrat} is a gathering algorithm for $n$ autonomous, indistinguishable and point-shaped connected robots in the Euclidean plane.
    A local robot model is used, where the robots especially only have a local vision and no compass.
    The gathering algorithm works in the $\mathcal{FSYNC}$ time model and runs in time $\calO(n^2)$.
    It is still unknown whether the model-specific time bound is tight for the general case, too.
    The second approach \cite{OptExactGatheringGrids2014} deals with gathering on a grid in the $\mathcal{ASYNC}$ time model.
    The robots have global vision and can move unbounded (finite) steps.
    The gathering strategy is optimal concerning the total number of movements.
    The third approach \cite{Jung2016} (accepted for IPDPS 2016, preliminary version \cite{Jung2015b}) is an asymptotically optimal gathering strategy for robots, connected as a closed chain on a grid, which is linear in the number of robots in the $\mathcal{FSYNC}$ time model.    
    Like the first approach \cite{gatheringthetanquadrat}, the robot model is local, especially only local vision and no compass.
    
	In this work, we use an idea from our gathering algorithm for a closed chain \cite{Jung2016}, yet drop the chain connectivity for sake of solving the general gathering on a grid under the same robot and time model.
    Our solution for the general gathering problem on the grid is a distributed synchronous algorithm that solves the problem asymptotically optimal in the total number of rounds.
    Note that the $\mathcal{FSYNC}$ time model is (like in \cite{Jung2016}) one of the main parameters that makes the construction and the analysis of an algorithm very challenging.
    The combination of preventing symmetry issues due the distributed robot behavior and constant length rounds require the presented complex solution.
    Contrary, if one would assume a fair scheduler in the $\mathcal{ASYNC}$ time model, which allows only one robot to be active at a time and finishes a round after every robot has been active at least once, a simple strategy could achieve the same $\calO(n)$ rounds.
    Specifically when compared to gathering on the Euclidean plane in the $\mathcal{FSYNC}$ time model, our runtime of $\calO(n)$ is a noteworthy result and beats the best known algorithm, which requires time $\calO(n^2)$ for gathering a swarm of $n$ robots.
    \paragraph{Our Local Grid Model}
        Our mobile robots need only simple capabilities:
        A robot moves on a two-dimensional grid and can change its position to one of its eight horizontal, vertical or diagonal neighboring grid points.
		A robot can see other robots only within a constant \emph{viewing radius} (measured in $L_1$-distance).
		We call the range of visible robots \emph{viewing range}.
        The robots controlled by our algorithm need $L_1$-distance \viewingradius, which can still be optimized. 
        The robots have no compass, no global control, no IDs and no global communication.
        A robot has a fixed small amount of memory to store a constant number of states.
		A robot can see the states of all robots inside the viewing range.

        Our algorithm uses the fully synchronous time model $\mathcal{FSYNC}$.
        Time is subdivided into equally sized rounds of constant lengths.
        In every round all robots simultaneously execute their operations in the common \emph{look-compute-move} model \cite{Cohen:2004a} which divides one operation into three steps.
        In the \emph{look step} the robot gets a snapshot of the current scenario from its own perspective, restricted to its constant sized viewing range.
        During the \emph{compute step}, the robot computes its action, and eventually performs it in the \emph{move step}.
	\paragraph{Outline of the Algorithm}
		The starting point for our problem is our strategy for the gathering of an arbitrary closed chain of robots on a grid \cite{Jung2016}.
		All operations to shorten the chain based on the strong definition of the robot's connectivity. 
		Each robot of a chain has exactly two well defined neighbors given by the input of the problem.

	    The main difference to the general gathering problem is that the connectivity of the chain is missing and so the algorithm could not be applied. 
        This makes general gathering more challenging than in the closed chain case.
		The idea of this paper is to find a substitute for the chain and the chain's connectivity in order to transmit the ideas of the algorithm.

		The substitute of the chain's connectivity is the following:
        We call two robots \emph{connected} if they are horizontal or vertical neighbors on the grid.
		This means that a robot is connected to at least one and at most four robots.
        Initially, under the restriction that the swarm is connected, the robots are arbitrarily distributed on the grid. 
        A swarm is connected, if the robots cannot be separated into two subsets such that no robot of the one subset is connected to any robot of the other subset and vice versa.

		In order to find a substitute for the chain we define the so called boundaries:
		The boundaries consist of all robots who have at least one unconnected side.
        The black robots of Figure~\ref{fig:boundary_ex} are part of the \emph{outer boundary} and the hatched robots are part of \emph{inner boundaries}.
		By that definition the swarm consists of one single outer boundary and typically multiple inner boundaries.
        A robot can detect if it is located on some boundary of the swarm, but because of its limited viewing range it does not know if it is located on the outer boundary or on some inner boundary.
        \begin{figure}[h]
            \centering
            \includegraphics[scale=\figscalee]{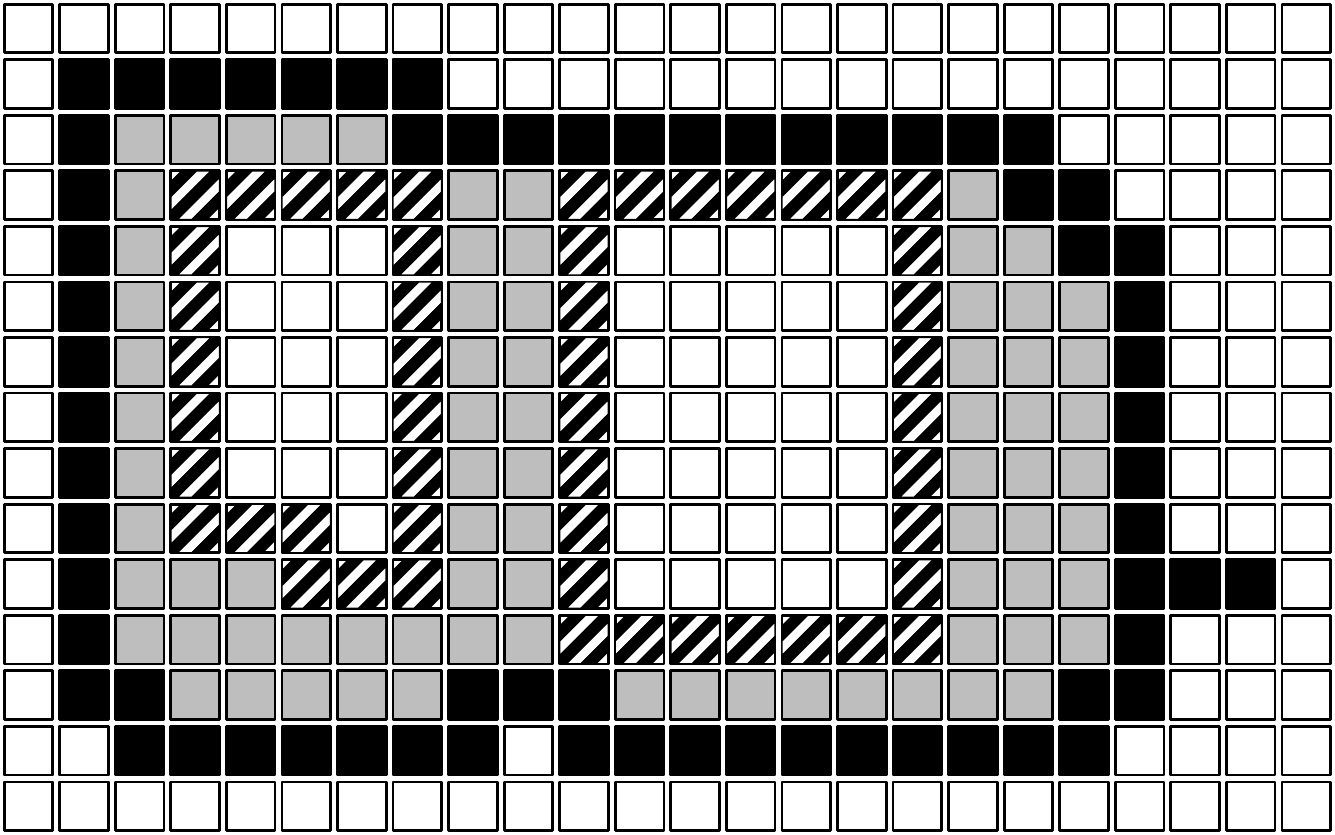}
            \caption{The hatched robots are part of \emph{inner boundaries} and the black robots are part of the \emph{outer boundary}.}
            \label{fig:boundary_ex}
        \end{figure}
		
        The outer boundary has the role as a closed chain and the goal of our algorithm is to shorten the outer boundary.
		However, the outer boundary does not have the same structure as the chain because the outer boundary can have multiple fringes.
		Thus our algorithm has to tackle two problems:
		First, we have to ensure that the algorithm will shorten the outer boundary.
		Second, we have to ensure that the shortening of the inner boundaries does not disturb the shortening of the outer boundary.
		Because the robots can not distinguish between inner and outer boundary our algorithm will also shorten the inner boundary.

        The main idea for our algorithm is the following:
        Our algorithm achieves progress in the gathering by performing \emph{merge operations}.
        A merge operation is a \emph{hop} of a robot located on the boundary onto the same grid cell as one of its neighbors and we remove one of them.
        A merge operation is only allowed if it does not disconnect the swarm.
        Sometimes no merge is possible without disconnecting the swarm.

		In that case we want the robots of the boundary to reshape the swarm.
		Reshapement means that the robots reach positions such that new merges are possible, which does not break the connectivity.
		Reshapement operations consist of two types: Starting a reshapement and continuing a reshapement.
		
		Starting a reshapement is allowed only to a robot that is located on a boundary and can see a certain configuration of robot positions in its viewing range.
		The robots that are allowed to start a reshapement generate a so called \emph{run state}, hop to another position, and give the run state to neighbors on the boundary.
		Robots receiving the run state are called \emph{runners} are allowed to continue the reshapement by hopping in following rounds and giving the state to neighbors of the boundary.
		In that way the run state will be moved along the boundary by the runners.
		If two runner converge and see each other inside their viewing range a merge can be performed.
        Gathering is finished when all robots are located within a $2\times 2$ square, since in our model that situation cannot be simplified anymore.
		In order to simplify the description of our algorithm, we sometimes allow a robot to reconfigure the positions of the robots in its viewing range.

        For counting the progress of our algorithm, it is sufficient to consider the merges of the outer boundary only, i.e. we look at the shortening of the outer boundary.
		After a constant number of rounds either a merge is possible or a new run state will be generated.
		A new generated run state will provide at least two runners that converge and enable a merge after $\calO(n)$ rounds.
 		The runners move in parallel and our strategy ensures that the runners do not disturb each other.
		Different started run states result in different merges.
        Therefore the algorithm gathers a swarm of $n$ robots in time $\calO(n)$.
        Our result is asymptotically optimal for worst-case swarms.
\section{Related work}
    There is a vast literature on robot problems, researching how specific coordination problems can be solved by a swarm of robots given a certain limited set of abilities.
    The robots are usually point-shaped (hence collisions are neglected) and positioned in the Euclidean plane.
    They can be equipped with a memory or are \emph{oblivious}, i.e., the robots do not remember anything from the past and perform their actions only on their current views.
    If robots are anonymous, they do not carry any IDs and cannot be distinguished by their neighbors.
    Another type of constraint is the compass model:
    If all robots have the same coordinate system, some tasks are easier to solve than if all robots' coordinate systems are distorted.
    In \cite{gathering-compasses,Izumi2012} a classification of these two and also of dynamically changing compass models, as well as their effects regarding the gathering problem in the Euclidean plane, is considered.
    The operation of a robot is considered in the \emph{look-compute-move model} \cite{Cohen:2004a}.
    How the steps of several robots are aligned is given by the \emph{time model}, which can range from an asynchronous $\mathcal{ASYNC}$ model (for example, see \cite{Cohen:2004a}), where even the single steps of the robots' steps may be interleaved, to a fully synchronous $\mathcal{FSYNC}$ model (for example, see \cite{localgathering}), where all steps are performed simultaneously.
    A collection of recent algorithmic results concerning distributed solving of basic problems like gathering and pattern formation, using robots with very limited capabilities, can be found in \cite{flocchinioverview}.
   
    One basic robot formation problem are the shortening and maintainance of a communication chain between two fixed endpoints.
    The robots then only have a local vision, no compass and a static connectivity is defined by neighborhoods along the chain.
    For this several results have been published in the $\mathcal{FSYNC}$ time model.
    The first shown runtime bound was $\calO(n^2\log(n))$ \cite{gtm}.
    Later, this has been improved \cite{hopper}:
    In the Euclidean plane, the \emph{Hopper} strategy delivers a $\sqrt{2}$-approximation of the shortest communication chain in time $\calO(n)$.
    Restricted to a grid, the \emph{Manhattan Hopper} strategy delivers an optimal solution in time $\calO(n)$.

    One of the most natural problems is to gather a swarm of robots in a single point.
    Usually, the swarm consists of point-shaped, oblivious, and anonymous robots. The problem is widely studied in the Euclidean plane.
    Having point-shaped robots, collisions are understood as merges/fusions of robots and interpreted as gathering progress \cite{MINCH}.
    In \cite{gathering-icalp} the first gathering algorithm for the $\mathcal{ASYNC}$ time model with multiplicity detection (i.e., a robot can detect if other robots are also located at its own position) and global views is provided.
    Gathering in the local setting was studied in \cite{localgathering}.
    In \cite{impossibilityofgathering} situations when no gathering is possible are studied.
    The question of gathering on graphs instead of gathering in the plane was considered in \cite{practicalrendevouzaktuell, rendezvousingraphen, gatheringOnRing}.
    In \cite{Stefano2013} the authors assume global vision, the $\mathcal{ASYNC}$ time model and furthermore allow unbounded (finite) movements.
    They show optimal bounds concerning the number of robot movements for special graph topologies like trees and rings.
    
    Concerning the gathering on grids,
    in \cite{gatheringongrids} it is shown that multiplicity detection is not needed and the authors further provide a characterization of solvable gathering configurations on finite grids.
    In \cite{OptExactGatheringGrids2014}, these results are extended to infinite grids, assuming global vision.
    The authors characterize \emph{gatherable} grid configurations concerning exact gathering in a single point.
    Under their robot model and the $\mathcal{ASYNC}$ time model, the authors present an algorithm which gathers \emph{gatherable} configurations
    optimally concerning the total number of movements.

    Assuming only local capabilities of the robots, esp.\ only local vision and no compass, makes gathering challenging.
    For example, a given global vision or alternatively just the knowledge of a global compass,
    the robots could compute the center of the globally smallest enclosing square or circle and just move to this point (global vision)
    or all robots without any local neighbors in front of them could simply move for example to the south-eastern direction
    and would finally meet (global compass).

    In the $\mathcal{FSYNC}$ time model, the total running time is a quality measurement of an algorithm.
    In this time model and additionally under the restriction, that the robots do not have a compass and have only local vision, i.e., they can only see other robots up to a constant distance instead of a global vision of the whole scenery, exist several results that prove runtime bounds.
    One of these are the strategies for shortening communication chains, we introduced above.
    I.e., \cite{gtm} that needs time $\calO(n^2\log(n))$ and \cite{hopper} that needs time $\calO(n)$ in the Euclidean plane as well as on a grid.

    Our recent result gathers a closed chain on a grid asymptotically optimal in linear time $\calO(n)$ \cite{Jung2016}, \cite{Jung2015b}.
    There, connectivity is statically given by a chain structure.
    For the more general gathering, i.e., using connectivity dynamically given by local vision in the Euclidean plane, an $\calO(n^2)$
    runtime bound has been shown \cite{gatheringthetanquadrat}.
    There, every robot synchronously computes the smallest enclosing circle only of the robots within its restricted viewing range and then moves towards its center.
    Repeatingly executing this synchronous behaviour finally solves the gathering.
    The authors also prove that for their algorithm the $\calO(n^2)$ bound is tight.
    For the problem itself, under this local model, a tight bound for the running time is still unknown.

    In the present paper, we tighten this bound to the asymptotically optimal value of $\calO(n)$ using the same time- and robot model, but for robots on a grid.
\section{The algorithm}
    In our model, merges/removal of robots mean progress of the gathering.
    Then, at the latest after $n-1$ removals the gathering is done.
    In our model, a robot can be removed, if two robots are located at the same grid cell.
    Then, one of them is removed.

    Because a robot's viewing range is restricted to just a constant size, these removals cannot be performed easily in general.
    Our algorithm performs two basic operations.
    More precisely, we have to deal with two cases:
    \begin{enumerate}
        \item Some neighboring robots perform a single hop such that afterwards at least two robots are located at the same position, while preserving the swarm connectivity. Then, we remove one of them.
        This is a so called \emph{merge} and further discussed in Subsection~\ref{ssec:merges}.\label{enum:easiershortening}
        \item If on some parts of the swarm merges are impossible, we perform so called reshapements to reshape the swarm in order to prepare merges in succeeding steps.
        They are explained in Subsection~\ref{ssec:runs}.
    \end{enumerate}
    In Subsection~\ref{ssec:merges}, we now start with case \ref{enum:easiershortening}.
    For the sake of simpler descriptions, we use terms like \emph{horizontal, vertical, downwards, left,\ldots}
    Since our robots do not have a common sense for this, the descriptions/figures are also to be understood in a mirrored or rotated manner.
    \subsection{Merges}\label{ssec:merges}
        A \emph{merge-operation}, or simpler a \emph{merge}, is the simultaneous local operation of a sequence of neighboring robots that merges at least one robot, while not harming the overall connectivity of the swarm.
        There are different merge operations, all parametrized by a parameter $k$, as depicted in Figure~\ref{fig:ALG_merge2}.
        In the figure, the case $k=1$ denotes the simplest variant, where only a single robot hops onto a grid cell occupied by another robot.
        Because the robot's viewing radius is $>1$, namely $\viewingradius$, also bigger merges are possible within the local vision.
        Then, for not to breaking the swarm's connectivity, multiple robots (the black ones in the figure) have to hop simultenenously:
        Using the local information about the occupied and empty cells in a viewing range, every robot can decide individually, based on its local knowledge, if it participates in the simultaneous merge operation or not.
        Specifically, the marked white grid cells must not contain any robot, any further marked cell must contain a robot, and the maximal size $k$ of a merge configuration is limited by the viewing radius.
        Moreover, all not explicitly depicted cells are ignored for the decision.

        The $k$ robots that constitute a merge form a subboundary of the swarm.
        When operating, these subboundary robots simultaneously hop one grid cell in the same direction (in the illustration, this means downwards).
        The swarm's connectivity is ensured, since a merge is only performed when the depicted white cells are empty.
        By requiring at least one grey cell, i.e., a robot that does not move, after the hops of the subboundary robots at least one robot from a grey cell will be located at the same cell as a robot from a formerly black cell and hence one robot is merged.
        \begin{figure}[h]
        \centering
            \includegraphics[scale=\figscale]{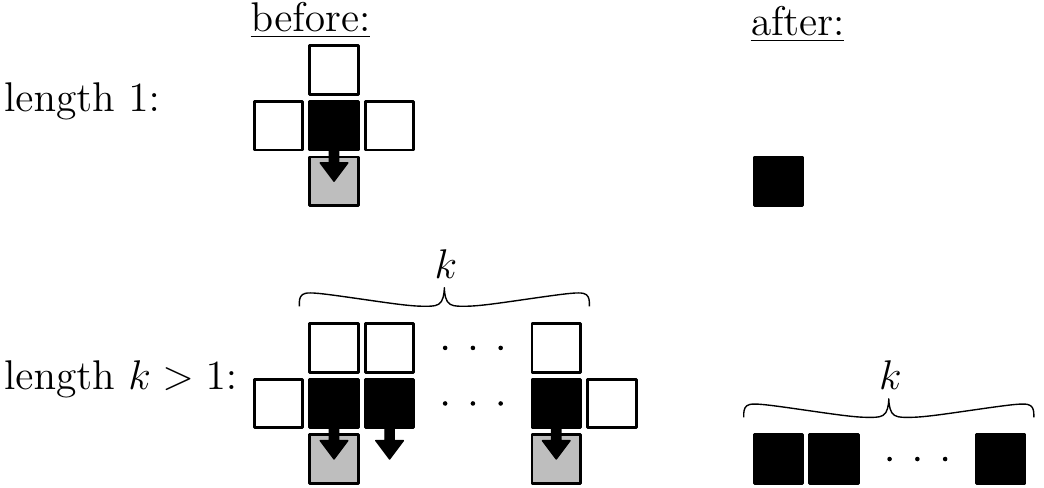}
            \caption{Robot subboundaries that allow progress hops (\emph{merge operations}). The value of the \emph{length} $k$ indicates the number of robots (the black ones) of the subboundary. $k$ is upper bounded by a robot's constant viewing radius. The operations are only performed, if the cells, marked by white squares, are empty. Else, the swarm's connectivity might break.}
            \label{fig:ALG_merge2}
        \end{figure}

        Possibly, several simultaneous merges can occur at different parts of the swarm.
        The problematic cases we have to discuss are those, where two of the robot subsets participating in the merge operations, i.e., the
        subboundaries, consisting of the black and grey robots of Figure~\ref{fig:ALG_merge2}, overlap (cf.\ Figure~\ref{fig:symmerges}).
        Precisely, two cases need a closer look (We refer to the black and grey robots of Figure~\ref{fig:ALG_merge2}.):
        \begin{enumerate}
            \item The beginnings and endings of the subboundaries overlap by two robots.\label{enum:oneoverlap}
            \item The beginnings and endings of the subboundaries overlap by three robots.\label{enum:twooverlap}
        \end{enumerate}
        An example for \ref{enum:oneoverlap}) can be seen in Figure~\ref{fig:symmerges}.$a)$:
        The robots of the subboundaries 1,2,3 all perform the hop as the black ones during the merge operation (Figure~\ref{fig:ALG_merge2}).
        The difference is that afterwards for example the robot $a$ is not located at the same position as $b$ and vice versa.
        So, no robot can be removed.
        But the outermost grey robots of Figure~\ref{fig:symmerges}.$a)$ do not move and there a merge occurs.

        An example for \ref{enum:twooverlap}) is shown in Figure~\ref{fig:symmerges}.$b)$.
        Here, the robot $r$ belongs to the both subboundaries $1$ and $2$.
        I.e, concerning subboundary $1$ $r$ would hop downwards, but concerning subboundary $2$ it would hop to the left.
        In this case, $r$ performs a diagonal hop to the lower left, while the other robots perform their usual hops.
        Afterwards, $r,a,b$ occupy the same grid cell and $a,b$ are removed without breaking the connectivity.
        \begin{figure}[h]
        \centering
            \includegraphics[scale=\figscale]{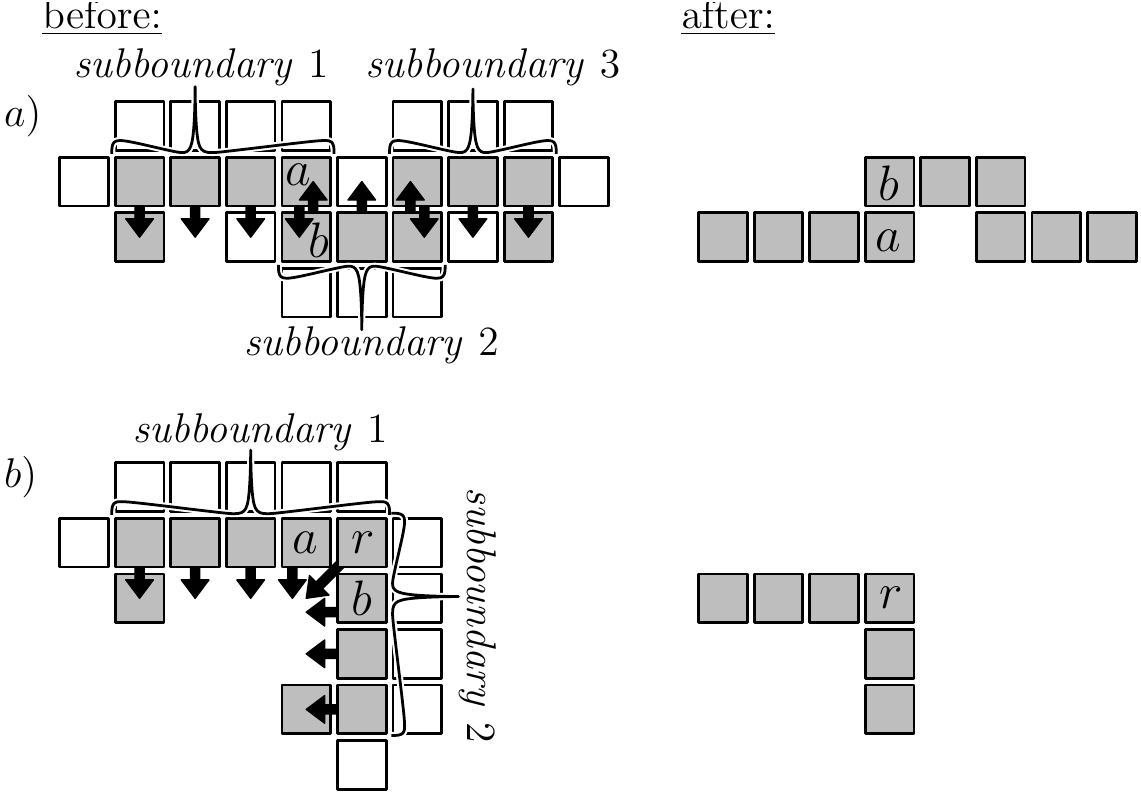}
            \caption{Overlapping merges. (significant examples)}\label{fig:symmerges}
        \end{figure}
    \subsection{Reshapement of the Swarm by Runners}\label{ssec:runs}
        From the global perspective, if globally nowhere in the swarm the algorithm's local merge operation is possible, then we call it a \emph{Mergeless Swarm}.
        In this case, our goal is to perform certain \emph{reshapements} of the swarm's outer boundary in order to make merge operations possible.
        Still considering the swarm from a global perspective, if there is a subboundary as depicted in Figure~\ref{fig:goodpair_ex_simple_norunners} by the black robots, then we continuously let an outermost robot of this subboundary perform the depicted diagonal hops.
        Eventually, this will shorten the subboundary enough to allow a (local) merge operation (cf.\ Figure~\ref{fig:ALG_merge2}).
        We say, these hops \emph{reshape} the subboundary.

        For an individual robot with its limited viewing range this raises some challenges.
        \begin{enumerate}
            \item When a robot decides to start the reshapement, within its restricted viewing range it does not know if its task of performing diagonal hops will lead to merges (In Figure~\ref{fig:goodpair_ex_simple_norunners}, this is the left outermost black one in round $i$.).
            \item In the following rounds, due to the local viewing ranges and skewed coordinate systems, the local view of an outermost black robot may be the same as that of its grey neighbor.
                So, we have to ensure that the reshapement is continued by the outermost black robots instead of by their grey neighbors.
            \item When robots decide to start the reshapement, symmetries of the subboundary's local shape may lead to breaking of the swarm's connectivity.
                In Figure~\ref{fig:symbreakrunstarts} this is the case if both $r$ and $r'$ start the reshapement.
        \end{enumerate}
        \begin{figure}[h]
        \centering
            \includegraphics[scale=\figscale]{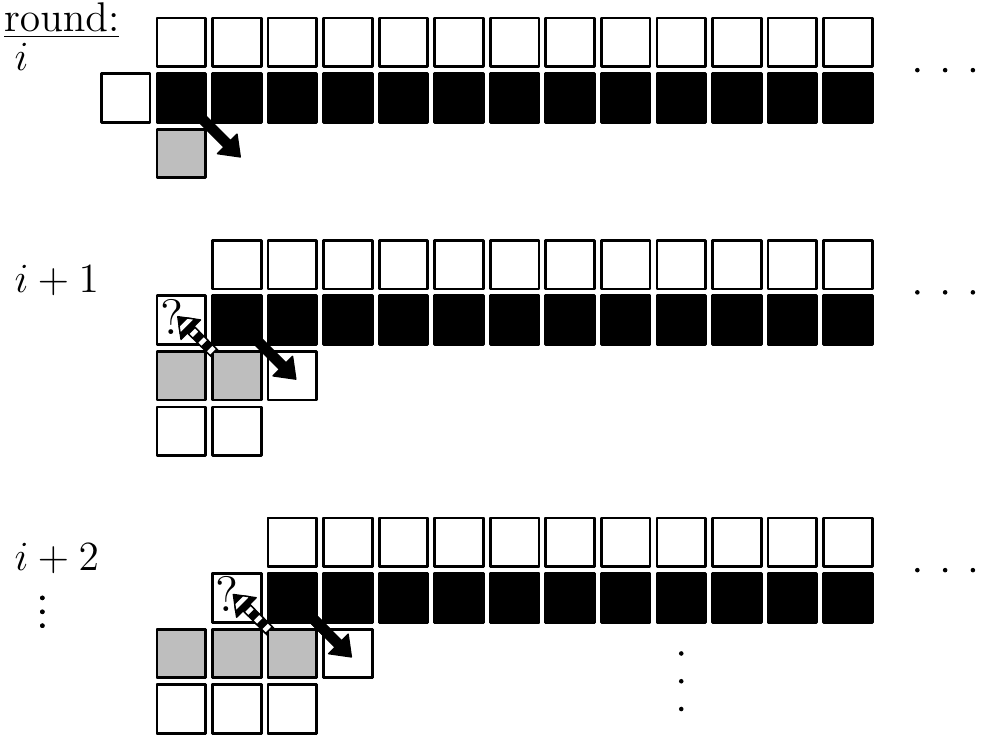}
            \caption{If the length of the black subboundary (i.e., the value of $k$ in Figure~\ref{fig:ALG_merge2}) is larger than the robots' viewing radius, we shrink it by letting one or both of the outermost black robots perform diagonal hops.
            We indicate such hops by diagonal arrows.}
            \label{fig:goodpair_ex_simple_norunners}
        \end{figure}
        \begin{figure}[h]
        \centering
            \includegraphics[scale=\figscale]{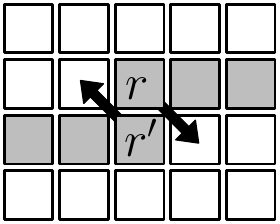}
            \caption{If $r$ and $r'$ both start reshaping the subboundary, the connectivity might break.}
            \label{fig:symbreakrunstarts}
        \end{figure}
        In the following, we tackle these challenges by introducing certain (local) states of a robot, which we call the \emph{run} state.
        We call a robot with an active run state a \emph{runner}.
        Robots can achieve this state in two different ways:
        \begin{description}
            \item [start runstate:] If the local subboundary within a robot's viewing range has a certain configuration, i.e., the relative positions of other robots and empty cells, then the robot decides on its own to generate the run state.
                We say, such a robot starts a run.
                Based on the configuration of the local subboundary, the run state gets a fixed moving direction along the boundary.
                A robot can start and store up to two run states at the same time.
                Figure~\ref{fig:runstartingrobots} shows how the local configurations must look like.
            \item [move runstate:] A runner $R(S)$ can move the run state $S$ to its boundary neighbor $r'$ in moving direction of $S$.
            We say, the run state has moved from $R(S)$ to $r'$, while its in ``start runstate'' initially set moving direction always remains unchanged. Afterwards, $r'$ is identified by $R(S)$.
        \end{description}
        Once a run state has been started in ``start runstate'', ``move runstate'' is executed in every of the following rounds.
        This means that the run moves along the boundary at constant speed and in the initially settled moving direction.

        When starting runs, the shapes of Start-A and Start-B in Figure~\ref{fig:runstartingrobots} ensure that the run starts cannot break the swarm's connectivity.
        For this, in a situation like in Figure~\ref{fig:symbreakrunstarts} we do not start any runs.
        We name subboundaries, consisting of shapes like in the figure, \emph{quasi lines}.
        Figure~\ref{fig:quasiline_ex2} gives an example of a quasi line.
        \begin{figure}[h]
            \centering
            \includegraphics[scale=\figscale]{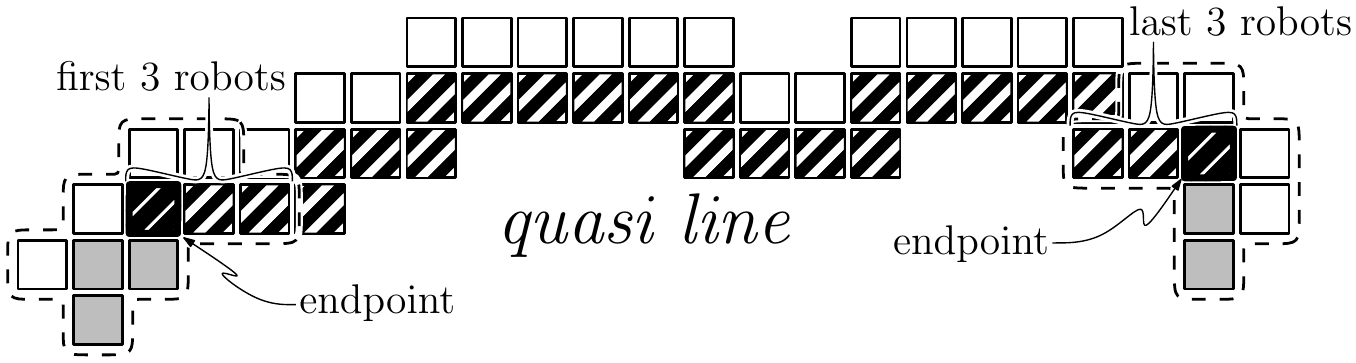}
            \caption{Example of a quasi line. The fat robots are its endpoints.}
            \label{fig:quasiline_ex2}
        \end{figure}
        \begin{definition}[quasi line]\label{def:quasiline}
            We call a subboundary a horizontal \emph{quasi line}, if the following points hold:
            \begin{enumerate}
                \item At least its first and last three robots are horizontally aligned.
                \item All its subboundaries of horizontally aligned robots contain at least three robots.
                \item All its subboundaries of vertically aligned robots contain at most two robots.
            \end{enumerate}
            In a Mergeless Swarm, at both ends of a quasi line a run starting subboundary Start-A or Start-B of Figure~\ref{fig:runstartingrobots} in a matching rotation or reflection occurs. (If the swarm is not mergeless, then the subboundaries outside the quasi line's endpoints may also have other shapes than these.)

            The definition of a vertical \emph{quasi line} follows analogously.
        \end{definition}
        \begin{figure}[h]
            \centering
            \includegraphics[scale=\figscale]{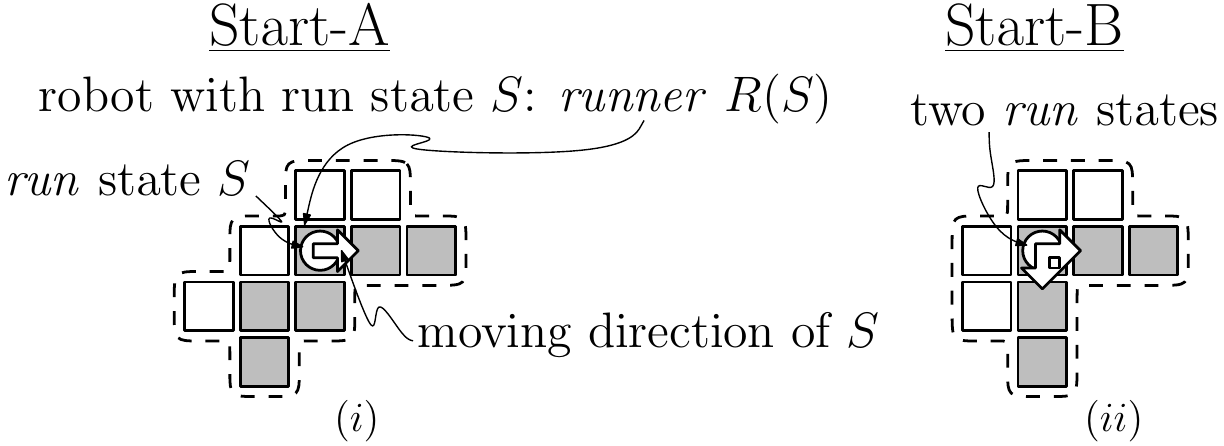}
            \caption{Run starting subboundaries: The robots marked by white circles decide to start the run states, only based on the relative positions of the marked robots and empty cells. Here, grey squares denote robots, while white squares denote cells that must be empty.
            The white arrows indicate the moving direction of the runs.
            This is the notation, we will use for marking a runner. $R(S)$ identifies the runner/robot which currently has the run state $S$;
            In $(ii)$, Start-B, the robot marked by the circle is the endpoint of a horizontally and a vertically aligned subboundary at the same time.
            Here, we must start two runs, moving in both directions along the boundary.}
            \label{fig:runstartingrobots}
        \end{figure}
        We let runs start at endpoints of quasi lines.
        We let them move along such quasi lines and, while doing this, perform reshapements of the boundary.
        For this, depending on the local shape of the subboundary, we require the run operations, shown in Figure~\ref{fig:ALG_hop_simple2}.
        \begin{figure}[h]
        \centering
            \includegraphics[scale=\figscale]{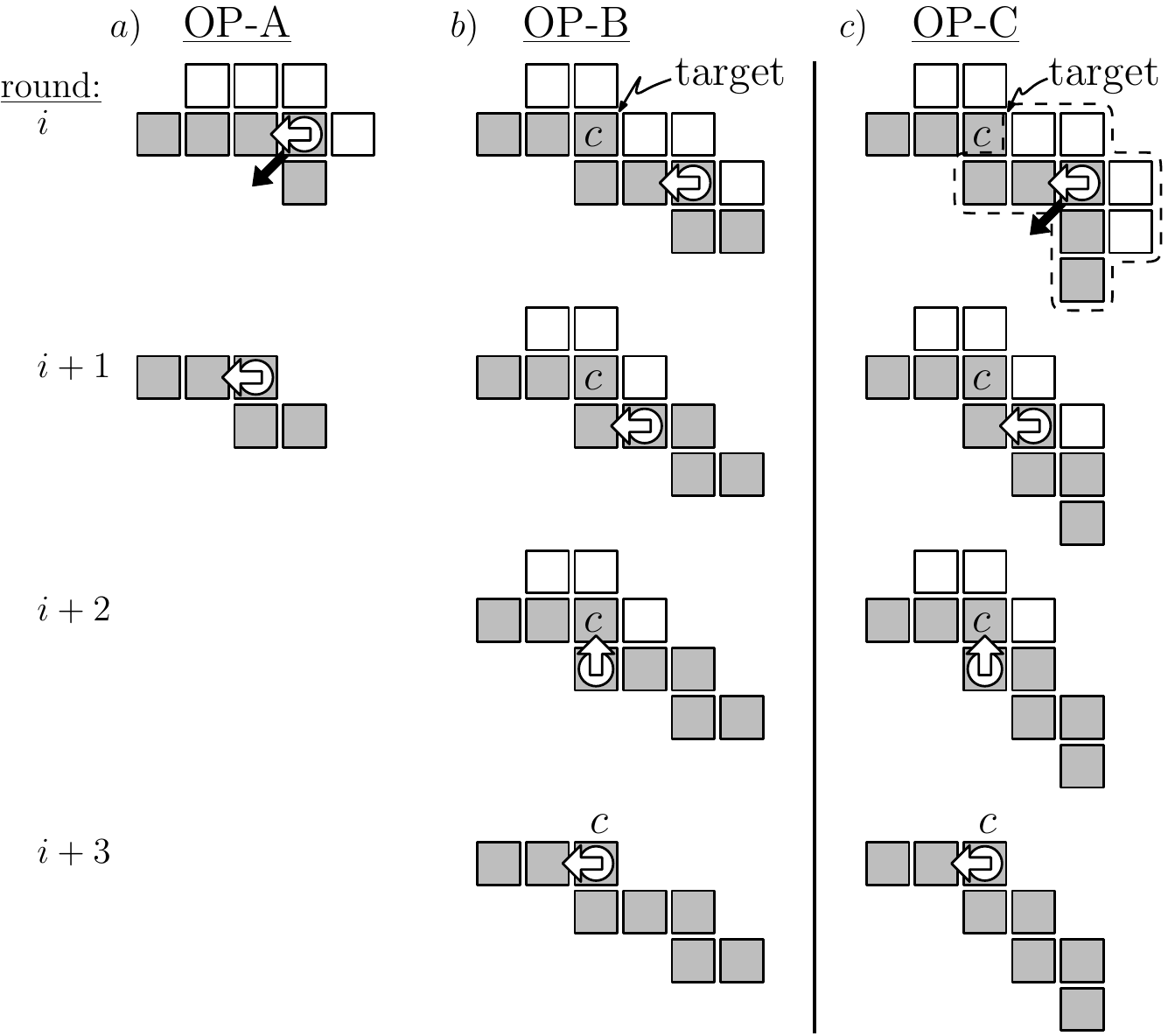}
            \caption{$a)$ OP-A: Reshapement by a runner. The operation takes only one round. $b)$ OP-B: No diagonal hops are performed until the target corner $c$ is reached. $c)$ OP-C: Special case when new runs start: First, perform one diagonal hop, then no diagonal hops until the target corner $c$ is reached.}
            \label{fig:ALG_hop_simple2}
        \end{figure}
        \begin{itemize}
            \item[$a)$] OP-A: The runner and at least the next 3 robots are located on a straight line. Here, the runner first performs a diagonal hop, then moves the run to the next robot.
            \item[$b)$] OP-B: The runner and only the next 2 robots are located on a straight line. Then, for 3 times the runners just move the run to the next robot without any diagonal hops. Afterwards, it is located at the target corner $c$.
            \item[$c)$] OP-C: This one is needed at most once for a new run, if started at the run starting subboundary Start-B (Figure~\ref{fig:runstartingrobots}.$(ii)$).
        \end{itemize}
        These operations shall finally enable merges.

        Two properties let multiple runs be active in parallel.
        \begin{enumerate}
            \item Usually, at different positions, the whole swarms contains multiple run starting subboundaries (cf.\ Figure~\ref{fig:runstartingrobots}) at the same point in time.
            \item No matter, if previously started runs are still active, every constant number of $L=\pipelininginterval$ rounds all robots simultaneously check if they can start new runs (cf.\ Figure~\ref{fig:runstartingrobots}) and if so, they do so.
        \end{enumerate}
        So, we have to deal with two runs, meeting on a quasi line.
        Then, we distinguish two cases (Cf.\ Figure~\ref{fig:runmeetingcases}.).
        \begin{figure}[h]
        \centering
            \includegraphics[scale=\figscale]{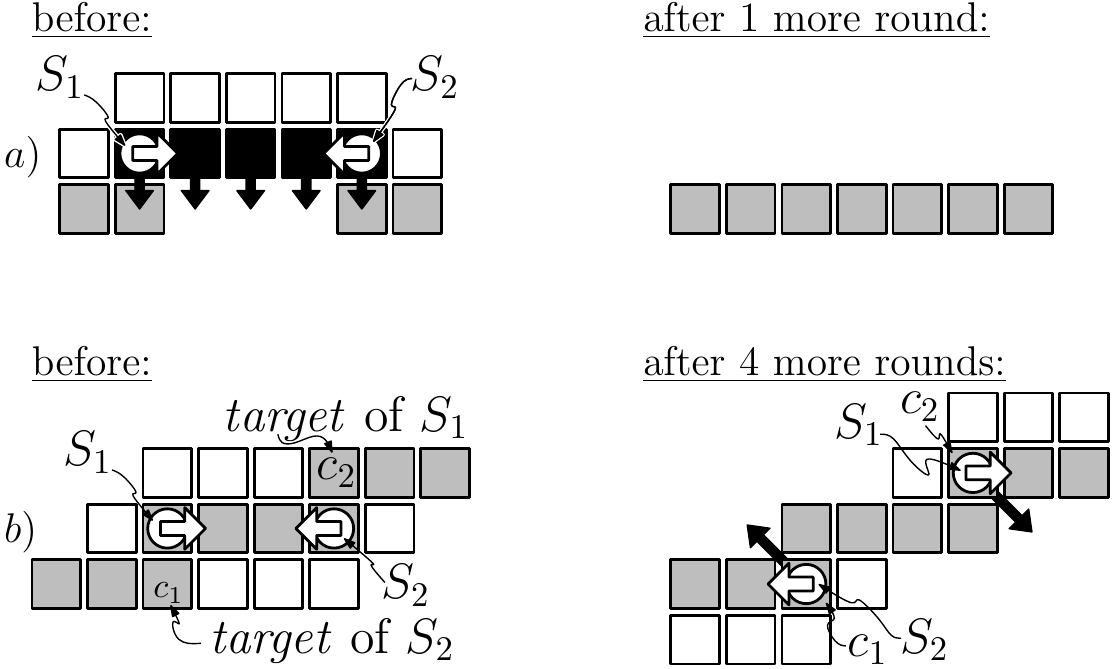}
            \caption{$a)$ \emph{Merge} operation: The reshapements of the runs $S_1,S_2$ have enabled a merge. This merge is performed and both are stopped. The operation takes only one round.
            $b)$ \emph{Run passing} operation: The runs $S_1,S_2$ cannot enable a merge.
            If their distance along the boundary is less or equal $\crossdistance$, then they pass each other by only keeping moving but without making the runners perform diagonal hops.
            Afterwards, i.e., when $S_1,S_2$ have reached their target robots/corners $c_2,c_1$, they return to normal operation.}
            \label{fig:runmeetingcases}
        \end{figure}
        $a)$ This run pair has enabled the desired merge.
        $b)$ This pair is oriented in a way that does not enable a merge.
        We let the runs of such pairs pass along each other.
        Figure~\ref{fig:runmeetingcases}.$b)$ shows for the case that both runs are located on the same quasi line how this is performed:
        At the time when their distance along the boundary (i.e., the number of robots on the subboundary connecting both $+1$) is $\crossdistance$ or less, they only keep moving along the boundary, but the runners do not perform reshapement hops.
        We call $\crossdistance$ the \emph{run passing distance}.
        This is repeated until $S_1$ is located at its target robot $c_2$ (then also $S_2$ is located at its target $c_1$).
        We call this the \emph{run passing} operation.
        Afterwards, the normal reshapement operations are continued.
        In Section~\ref{sec:runpassingcomplete}, we provide a full explanation of all possible situations of needed run passing operations.
        While in the easy case a viewing radius of $\viewingradiusident$ is the minimum value that suffices, we then need the unoptimized value of $\viewingradius$.
    \subsection{Stopping Runs}
        \begin{figure}[h]
        \centering
            \includegraphics[scale=\figscale]{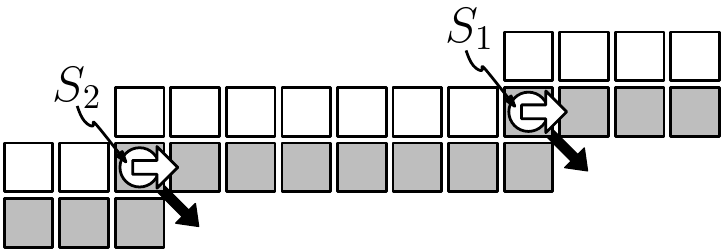}
            \caption{The \emph{distance} between $S_1,S_2$ is $8$.}
            \label{fig:runtermsnew}
        \end{figure}
        We use the example of Figure~\ref{fig:runtermsnew} for introducing some terms when speaking about runs:    
        Because $S_1, S_2$ are moving in the same direction, we call them \emph{sequent} runs, while relative to their moving direction
        $S_1$ is located \emph{in front of} $S_2$.
        The \emph{distance} between them is defined as the number of robots on the subboundary connecting both $+1$.
        We say, $S_1$ is \emph{visible} to $S_2$, if the \emph{distance} between them is $\leq \viewingradius$.

        In order to ensure that runs, being active at the same time, can work correctly, some of them have to stop in certain situations.
        More precisely, we let a runner $R(S)$ stop/terminate its run $S$, if one of the conditions of Table~\ref{table:runterminate} is true.
        \begin{table}
            \fbox{\parbox{0.97\columnwidth}{
                A runner stops/terminates its run, if at least one of the following conditions is true:
                \begin{enumerate}
                    \item It can see the next sequent run in front of it (This happens if sequent runs have come too close to each other, e.g., because of merge operations.).\label{enum:stoptooclose}
                    \item It can see the quasi line's endpoint in front of it.\label{enum:stopendpoint}
                    \item It was part of a merge operation.\label{enum:stopmerge}
                    \item While it performs the run passing operation (Figure~\ref{fig:runmeetingcases}.$b)$), 
                    some change of the subboundary's shape prevents the operation from being completed successfully. (This can happen because of a merge.)\label{enum:stoppass}
                    \item While it performs the operation OP-B or OP-C (Figure~\ref{fig:ALG_hop_simple2}), some change of the subboundary's shape prevents the operation from being completed successfully. (This can happen because of a merge.)\label{enum:stoplongop}
                    \item While it performs the operation OP-A or OP-C, it has hopped onto an occupied cell. (Then, one of both robots is removed.)
                \end{enumerate}
            }}
            \caption{Conditions which let a run terminate.}
            \label{table:runterminate}
        \end{table}

        We give some more detailed explanations concerning some of the conditions of Table~\ref{table:runterminate}:
        \begin{itemize}
            \item[\ref{enum:stoptooclose})]
                For example, because of merges two sequent runs may come too close to each other, which then might hinder the pipelining.
                The affected runners can detect this on their own.
                The criterion for this is that the next sequent run in front of them becomes visible.
                Then, the termination condition \ref{enum:stoptooclose}) matches and the run behind stops.
            \item[\ref{enum:stoppass})]
                We assume, that in Figure~\ref{fig:runmeetingcases}.$b)$ to the right of robot $c_2$ another run $S_3$, moving in the same direction as $S_2$, is located.
                During the rounds in which $S_1$ and $S_2$ are performing their run passing operation, $S_3$ keeps moving towards $c_2$.
                Now, for example, it may happen that because of the reshapements of $S_3$, $c_2$ becomes part of a merge operation.
                Then, $c_2$ would hop downwards such that the corner shape does not exist anymore.
                Because this corner has been the target of $S_1$, $S_1$ could not continue its reshapements after the run passing, so it terminates.
        \end{itemize}
        Summarizing, all robots synchronously execute the algorithm, shown in Figure~\ref{fig:algo}.
        \begin{figure}[h]
            \fbox{\parbox{0.97\columnwidth}{
                Every robot $r$ every round checks the following three steps:
                \begin{enumerate}
                    \item \underline{Merge:} If $r$ detects a possible merge within its viewing range then
                        \begin{itemize}
                            \item if $r$ is one of the black robots in Figure~\ref{fig:ALG_merge2}, it hops downwards.
                            \item if afterwards $r$ is located at the same position as one of the grey robots, the grey one is removed without breaking the connectivity.
                        \end{itemize}
                    \item \underline{Run Operations:} If $r$ is a \emph{runner}, then
                        \begin{enumerate}
                            \item Its run terminates/stops if any of the conditions of Table~\ref{table:runterminate} is true.
                            \item Runner's Movement and Reshapement
                                \begin{itemize}
                                    \item Run passing:
                                        \begin{itemize}
                                            \item If $r$ is currently in progress of executing the run passing operation (Figure~\ref{fig:runmeetingcases}.$b)$), then this operation is continued.
                                            \item Else, if $r$ can see a run in front of it, such that both are moving towards each other and the distance between them is less or equal than the \emph{run passing distance}, then $r$ starts the run passing op.
                                        \end{itemize}
                                    \item If $r$ is not in progress of passing, then
                                        \begin{itemize}
                                            \item If $r$ is in progress of executing a run operation OP-B or OP-C (cf.\ Figure~\ref{fig:ALG_hop_simple2}), which take more than one round, then this one is continued.
                                            \item Else: $r$ executes the matching new run operation OP-A, OP-B or OP-C.
                                        \end{itemize}
                                \end{itemize}
                        \end{enumerate}
                    \item \underline{Start new runs:} Every $(L=\pipelininginterval)$th round, $r$ checks if it can start a new run:\\
                        If $r$ is one of robots marked by circles, of the run starting subboundaries Start-A or Start-B (cf.\ Figure~\ref{fig:runstartingrobots}), then it starts one resp.\ two runs.
                \end{enumerate}
            }}
            \caption{The algorithm.}
            \label{fig:algo}
       \end{figure}

\section{Why the strategy produces progress in gathering}
    For measuring progress, we only consider runs, started pairwise at both ends of certain quasi lines.
    The kind of runs pairs which in any case lead to a merge, are called \emph{good pairs}.
    \subsection{Good pairs}\label{ssec:goodpairs}
        Assume a newly started run pair, connected by a quasi line.
        We call this pair a \emph{good pair}, if the following is true (cf.\ Figure~\ref{fig:goodpair_ex1}):
        \begin{figure}[h]
        \centering
            \includegraphics[scale=\figscale]{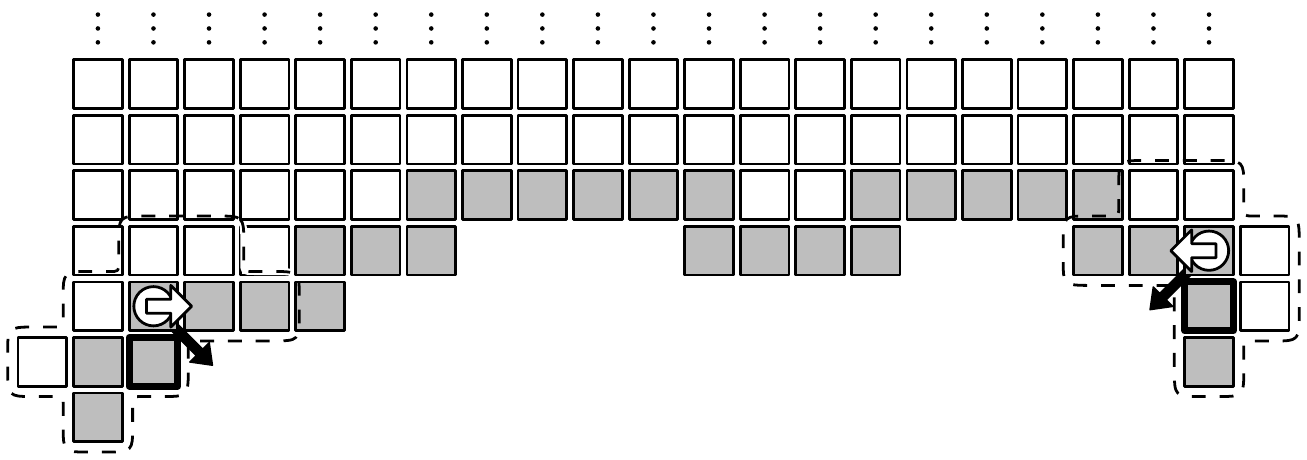}
            \caption{Pair of runs, connected by a quasi line. The runs are a \emph{good pair} if the outer subboundary neighbors (the fat robots) are both located downwards and the globally the whole area above does not contain any robots.}
            \label{fig:goodpair_ex1}
        \end{figure}
        The exterior neighbors (the fat bordered robots in the figure) of the run pair are both located downwards
        and the whole (global) area above the connecting quasi line does not contain any robot
        (The same definition analogously holds for any rotation.).
        Looking at our run operations OP-A, OP-B and OP-C as well as the run passing operation and merges, we notice that
        none of them could let a robot which is not part of the considered quasi line hop into this empty area.

        Such good pairs always enable a merge if they have been moving close enough together and then terminate.
        Figure~\ref{fig:runmeeting_simple} shows an example of a good pair where the shape of the quasi line connecting both, is a straight
        line.
        The example of the figure shows that using runs, we solve the problem, we have had in Figure~\ref{fig:goodpair_ex_simple_norunners}, where it was impossible to locally decide by the robots, which of them has to perform the reshapement hop.
        \begin{figure}[h]
        \centering
            \includegraphics[scale=\figscale]{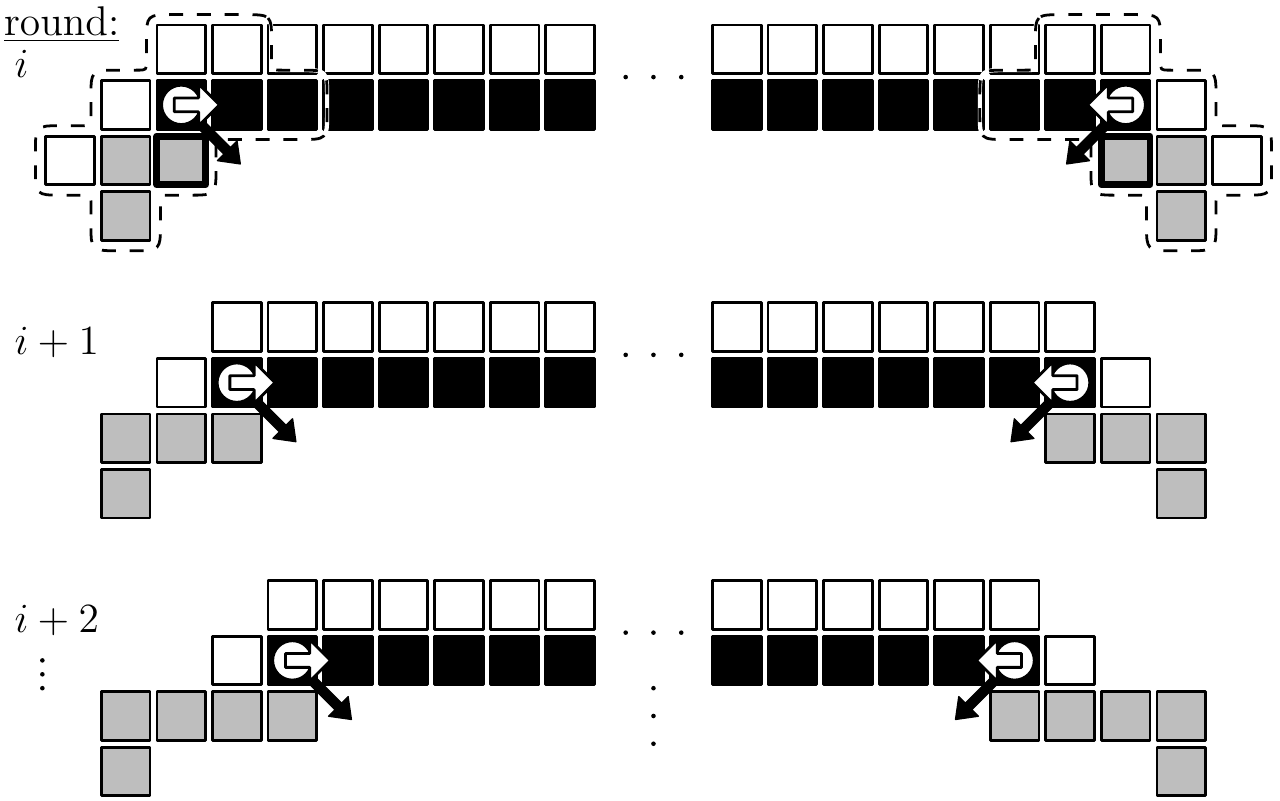}
            \caption{A \emph{good pair} of runs. In round $i$, the new runs start (cf.\ Figure~\ref{fig:runstartingrobots}).
            Afterwards, the action OP-A of Figure~\ref{fig:ALG_hop_simple2} are repeatingly executed.
            The runs are moving closer and closer together.
            If the runs have moved close enough, a merge can be performed.}
            \label{fig:runmeeting_simple}
        \end{figure}
        \begin{figure}[h]
            \centering
            \includegraphics[scale=\figscale]{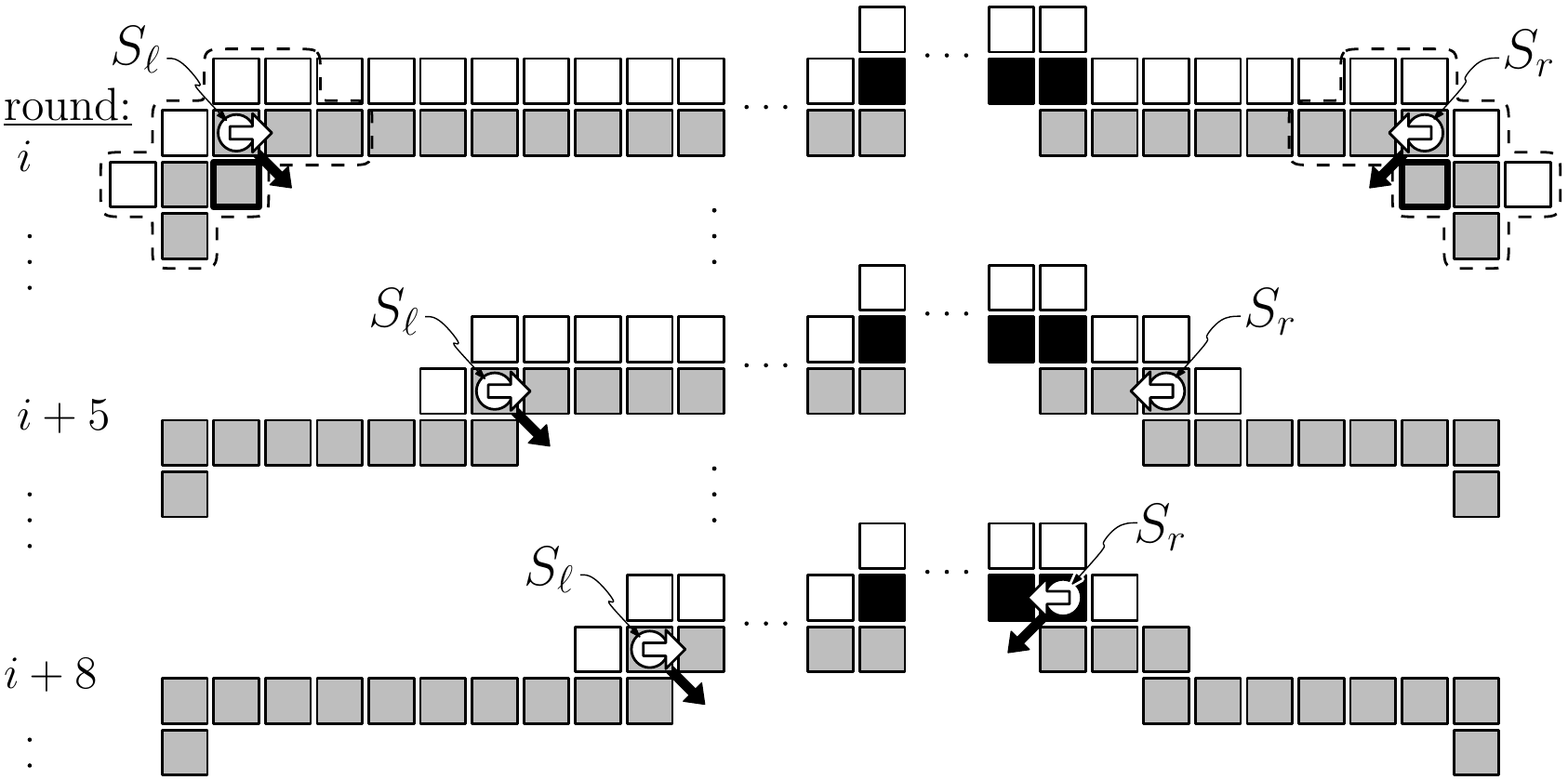}
            \caption{A good pair, started on a quasi line in order to shorten the black subboundary. Several and different run operations are needed.}
            \label{fig:goodpair_ex2_quasiline}
        \end{figure}
        Figure~\ref{fig:goodpair_ex2_quasiline} shows an example for a more general case of a quasi line:
        Again, the black subboundary will be reshaped for performing a merge.
        But now, several run operations (cf.\ Figure~\ref{fig:ALG_hop_simple2}) are needed until a runner arrives at and endpoint of the black subboundary:
        Until round $i+4$, the runners execute only the operation OP-A (cf.\ Figure~\ref{fig:ALG_hop_simple2}).
        Afterwards, the runner $R(S_r)$ starts operation OP-B, while $R(S_\ell)$ still executes OP-A.
        After operation OP-B has been processed completely, $S_r$ finally is located at the right end of the black subboundary (round $i+8$).
        Now, its runners start shortening this subboundary by executing OP-A until the merge can be performed.
        Afterwards, $S_\ell$ is still active and keeps moving and will stop at the latest when an endpoint of the quasi line
        becomes visible.
        
        Many of the started runs do not belong to good pairs.
        In Section~\ref{sec:correctness}, we will prove the actual existence of good pairs.
    \subsection{Pipelining}
        We will show that even if multiple good pairs are nested into each other, different good pairs will enable different merges.
        This is what we call \emph{pipelining}.
        \begin{figure}[h]
        \centering
            \includegraphics[scale=\figscale]{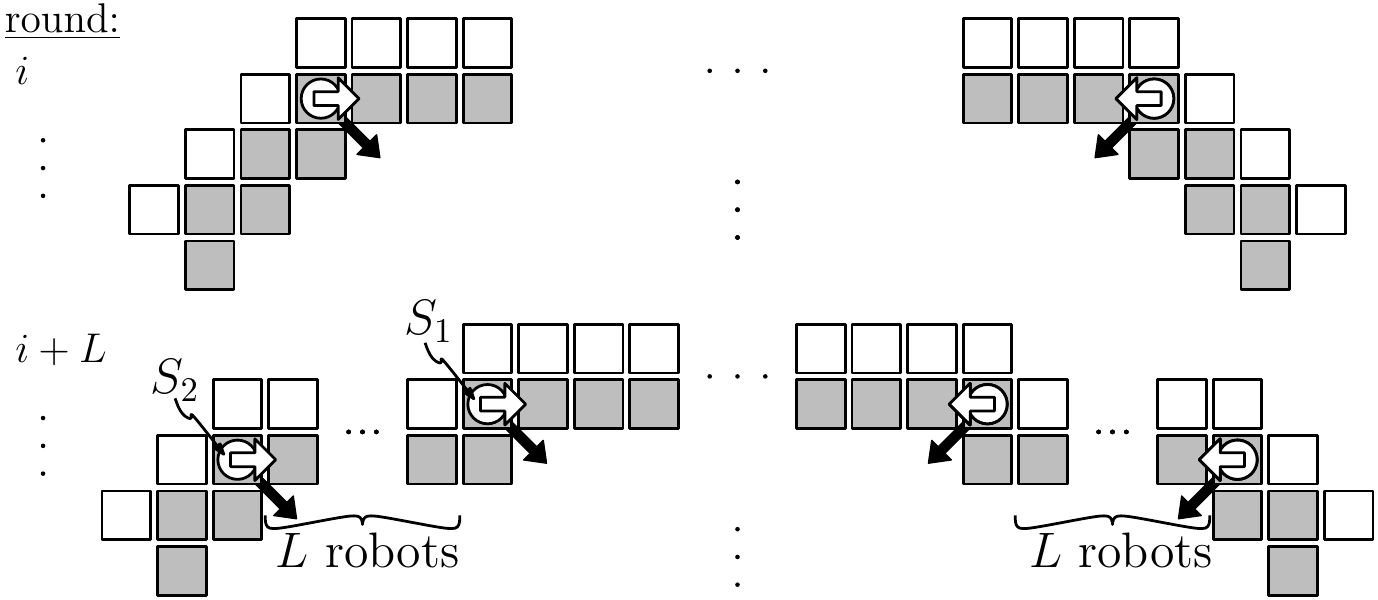}
            \caption{Pipelining of runs. New runs are started every $L=\pipelininginterval$ rounds.}
            \label{fig:pipelining_ex}
        \end{figure}
        Figure~\ref{fig:pipelining_ex} shows an example for the case that the quasi line has the shape of a straight line.
        The more general case works analog to this.

        As runs are moving with constant speed, the runs of the inner good pair will meet and as the result enable a merge and stop, first.
        Then, obviously the outer good pair will also enable a (different) merge, some rounds later.
\section{Correctness and running time}\label{sec:correctness}
        Our basic argumentation for the total running time is the following:
        If every round a merge can be performed, then the time needed for the gathering is obviously upper bounded by $n$,
        with $n$ being the number of robots.
        If no merge can be performed, then every $L=\pipelininginterval$ rounds a \emph{good pair} can be started which reshapes the
        swarm such that a merge can be performed after at most $n$ further rounds.
        As $L$ is a constant, this then leads to a linear total running time $\calO(n)$.
        
        In this section, we give the formal proofs.
        Our goal is the proof of the correctness and the linear running time (Theorem~\ref{thm:runningtime}).
        In our main approach, we want to enable merges by good pairs if else no merge could be performed.
        The following two lemmas prove that this actually works.
        They are the base for the proof of the theorem.
        The proofs of the lemmas can be found in Subsection~\ref{ssec:progresspair_exists} and \ref{ssec:progresspair_merge}.
 
        \paragraph{Progress pairs}
        Because of the local vision, new runs and maybe also good pairs are started every $L=\pipelininginterval$ rounds, regardless of whether or not at somewhere in the swarm merges can be performed.
        In the following analysis, we will only argue with new good pairs which are started if during the last $L-1$ and the current round in the whole swarm no merge has been performed.
        We distinguish such good pairs from others by calling them \emph{progress pairs}.
        \begin{lemma}\label{lem:progresspair_exists}
            Every $L=\pipelininginterval$ rounds either a merge has been performed 
            or else a new \emph{progress pair} is started.
        \end{lemma}
        \begin{lemma}\label{lem:progresspair_merge}
            For all \emph{progress pairs} the following properties hold.
            \begin{enumerate}[a)]
                \item Every progress pair enables a merge 
                (after at most $n$ rounds).\label{enum:progpairmerge}
                \item Different progress pairs enable different merges.\label{enum:progpairuniqueness}
            \end{enumerate}
        \end{lemma}
        Using these two lemmas, we can now prove the total linear running time.
        \begin{theorem}\label{thm:runningtime}
            Given a swarm of $n$ robots. Then, after $\calO(n)$ rounds gathering is done.
            This is asymptotically optimal.
        \end{theorem}
        \begin{proof}
            We subdivide time into intervals of lengths $L$, where $L$ denotes a number of rounds.
            Merges can be performed during at most $n$ such intervals, because every merge removes at least one robot.
            In all other intervals a new progress pair starts (Lemma~\ref{lem:progresspair_exists}).
            Each of these progress pairs leads to a merge (Lemma~\ref{lem:progresspair_merge}.$\ref{enum:progpairmerge})$).
            Because no two of them lead to the same merge (Lemma~\ref{lem:progresspair_merge}.$\ref{enum:progpairuniqueness})$),
            the number of intervals without merges is also upper bounded by $n$.

            By Lemma~\ref{lem:progresspair_merge}, a progress pair needs at most $n$ rounds until it has led to a merge.
            We assume the worst-case, which is that in the last of the $2\cdot n$ intervals the last progress pair was started.
            Then the total running time is upper bounded by $2n\cdot L\; +n$, which proves the upper bound of the theorem, because $L$ is a constant.

            In our model, the diameter of the initial configuration provides the worst-case lower bound $\Omega(n)$ for any gathering strategy.
        \end{proof} 
    \subsection{Proof of Lemma~\ref{lem:progresspair_exists}}\label{ssec:progresspair_exists}
        \begin{proof}
            The lemma requires that the swarm is a \emph{Mergeless Swarm}.
            As the basis of our proof, we start by showing that then the outer boundary of the swarm only consists of quasi lines and stairways (cf.\ Figure~\ref{fig:stairway_ohne-runs}).
            First, we assume that every robot is connected to only two other robots.
            Then, especially no robots are located at the inside of the swarm, and the swarm consists only of its \emph{outer boundary}.
            We start with the assumption that merges are only possible up to the length $2$ (cf.\ Figure~\ref{fig:ALG_merge2}).
            Then, as by definition all horizontal subchains of a horizontal quasi line consist of at least $3$ robots, no merge can be performed on them.
            In order to connect two horizontal or two vertical quasi lines without enabling a merge, they must be connected by so called \emph{stairways}.
            Stairways can have arbitrary length and are subchains of alternating left and right turns.
            Figure~\ref{fig:stairway_ohne-runs} shows an example.
            The bicolored robots are the connection points between the quasi lines and the stairway.
            \begin{figure}[h]
                \centering
                \includegraphics[scale=\figscale]{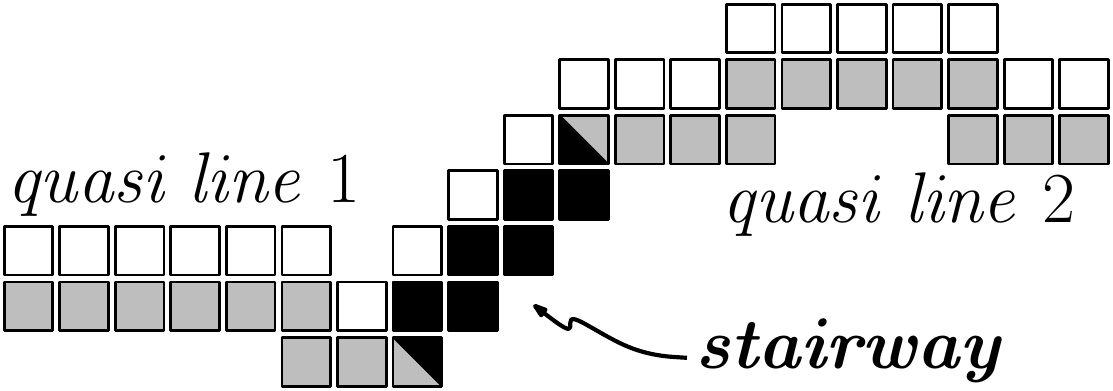}
                \caption{Two quasi lines, connected by a \emph{stairway}.}
                \label{fig:stairway_ohne-runs}
            \end{figure}
            All differently shaped connecting subchains would allow merges.
            So, assuming that every robot is connected to only to other robots, the whole outer boundary consists only of quasi lines and
            stairways.
            This clearly is also the case, if longer merges are allowed as, for horizontal quasi lines,
            increase the number of horizontally aligned robots but decrease the number of vertically aligned ones (cf.\ Definition~\ref{def:quasiline} and Figure~\ref{fig:ALG_merge2}).

            Now, we also allow connectivities bigger than $2$.
            Then robots may also be located at the swarm's inside or some parts of the outer boundary may be neighboring (i.e., if at that position the diameter~$=2$).
            Such robots can hinder merges.
            \begin{figure}[h]
                \centering
                \includegraphics[scale=\figscale]{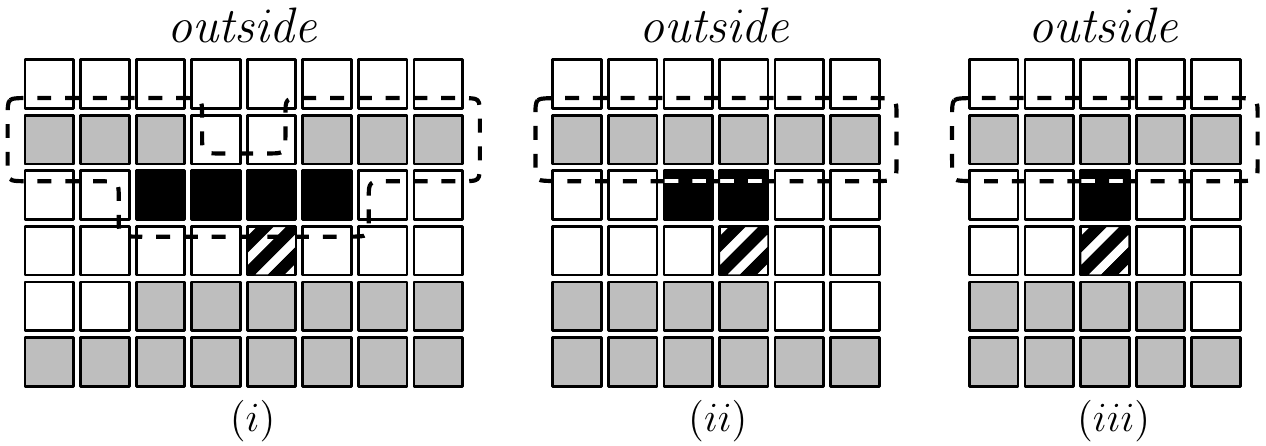}
                \caption{The hatched robot prevents the merge.}
                \label{fig:inside_merges}
            \end{figure}
            Figure~\ref{fig:inside_merges} shows examples for configurations with existing inside robots.
            In the figure, the dashed bordered robots are part of the outer boundary.
            The hatched robots prevent the black ones from performing a merge for which the black ones needed to hop upwards.
            So, although the swarm is mergeless, the outer boundary may still contain mergable shapes.
            For merges of length $>2$ ($(i)$ in the Figure), the outer boundary then still is a quasi line.
            For merges of length $1$ and $2$ ($(i,ii)$ in the Figure), the black robots are not part of the outer boundary.
            So, in every \emph{Mergeless Swarm}, the outer boundary only consists of quasi lines and stairways.

            Now, we can show that in a \emph{Mergeless Swarm} always a progress pair exists.
            For this, we assume a global north and west and construct a vector chain along the swarm's outer boundary as follows.
            We consider the upper envelope of the swarm and take its left- and rightmost robots $s$ respectively $t$.
            Because $s,t$ must be located on the boundary of the swarm's smallest enclosing rectangle and the swarm is mergeless,
            below each of them must exist two additional robots (cf.\ Figure~\ref{fig:upper-boundary_def_by_ex_1}).
            We start from the left at the vector $\vec{v_0}$ constructing a vector chain in clockwise orientation along the outer boundary
            of the swarm, ending at the vector $\vec{v_m}$.
            The figure shows a significant example of this construction.
            Here, the grey colored robots denote the corresponding part of the swarm's outer boundary.
            \begin{figure}[h]
                \centering
                \includegraphics[scale=\figscale]{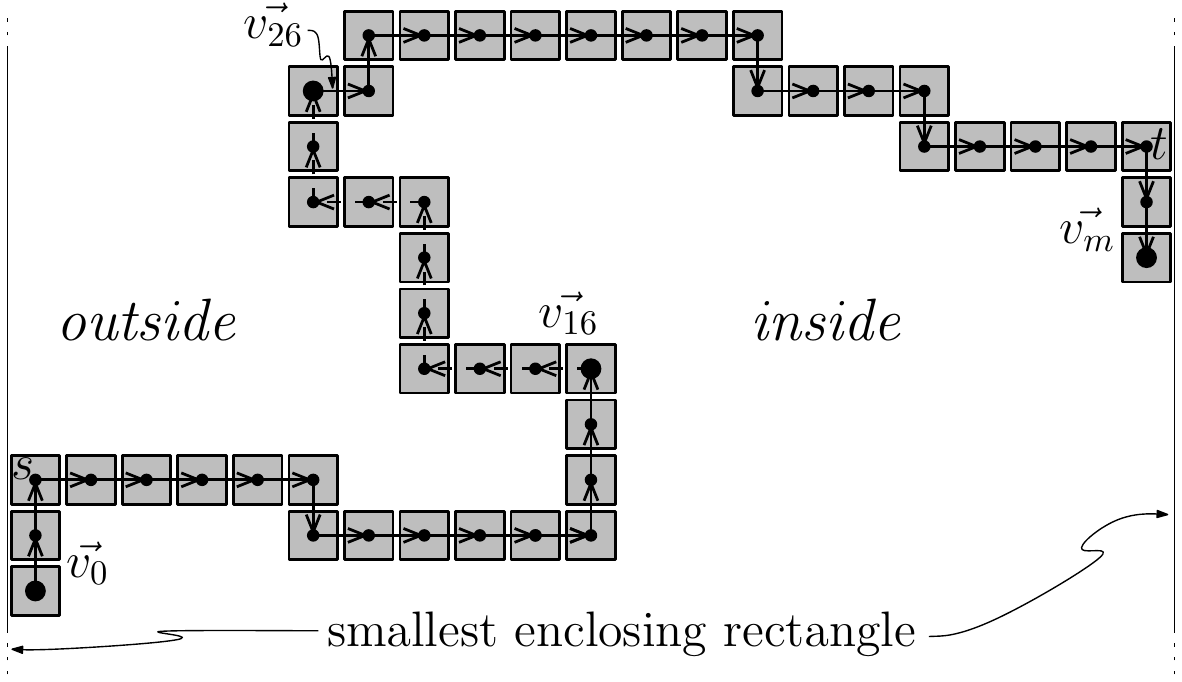}
                \caption{Constructing a vector chain along the swarm's outer boundary.}
                \label{fig:upper-boundary_def_by_ex_1}
            \end{figure}
            Starting from its left end, we divide this vector chain into longest $x$-monotone subchains.
            The vector, originated at $s$, by construction points to the east.
            The second subchain starts when the first vector points to the west.
            The third subchain starts when the first vector points to the east again,
            and so forth.
            In the figure, these are the vectors $\vec{v_{16}}$ respectively $\vec{v_{26}}$.
            The constructed vector chain may overlap itself at places where the diameter of the swarm's boundary amounts only $1$, but
            cannot contain any crossings.
            Because all constructed subchains are $x$-monotone, at least one of them is fully part of the upper envelope.
            Let $\mathcal{C}$ denote this subchain.
            If we consider $\mathcal{C}$, then, as the swarm is mergeless, both predecessors of its first vector must point to the north, i.e., the corresponding robots are located downwards.
            By the same arguments, both successors of its last vector must point to the south, i.e., the corresponding robots are also located downwards.
            This also holds if $s$ or $t$ the endpoints of this vector chain.
            Then, on both sides are three neighboring robots on a vertical straight line. 
            By Definition~\ref{def:quasiline} of a quasi line, these robots must be part of vertical quasi lines.

            In contrast, if looking at the subchain $\mathcal{C}$, then because the swarm is mergeless, $\mathcal{C}$ must contain at least two succeeding vectors, both pointing to the east.
            The three connected robots then must be part of a horizontal quasi line.
            At the transition between horizontal and vertical quasi lines, possibly connected by stairways, the shape matches the
            starting subboundaries Start-A or Start-B of Figure~\ref{fig:runstartingrobots}.
            As both vertical quasi lines are located downwards w.r.t.\ the horizontal quasi line of $\mathcal{C}$, a good pair can be started (cf.\ Subsection~\ref{ssec:goodpairs}).
            Because we have assumed that the swarm has been a Mergeless Swarm also during the previous $L-1$ rounds, this good pair is a progress pair.
        \end{proof}
        \subsection{Proof of Lemma~\ref{lem:progresspair_merge}}\label{ssec:progresspair_merge}

            For the proof, we need the run invariants of Lemma~\ref{lem:runinvariants}.
            \begin{lemma}\label{lem:runinvariants}
                The value $L=\pipelininginterval$ and the value $\viewingradius$ for the viewing radius ensure, that for every run $S$ until it terminates, the following invariants holds.
                \begin{enumerate}
                        \item Every round, $S$ moves one robot further in moving direction.\label{enum:runinvmovement}
                        \item After the first three rounds after its start the reshapements of the runner $R(S)$ do not violate the
                              quasi line Definition~\ref{def:quasiline} of its own quasi line.\label{enum:runinvquasiline}
                        \item The reshapements of $R(S)$ do not violate the quasi line shape of other quasi lines if they belong to good pairs.\label{enum:runinvforeignquasilines}
                        \item $S$ cannot see other sequent runs in front of it.\label{enum:runinvvisseqrun}
                        \item $S$ is either in progress of passing along another run or the runner $R(S)$ executes one of the
                              operations OP-A, OP-B or OP-C (Figure~\ref{fig:ALG_hop_simple2}).\label{enum:runinvops}
                        \item Good pairs stay being good pairs.\label{enum:runinvgoodpairs}
                \end{enumerate}
            \end{lemma}
            Now, we can prove Lemma~\ref{lem:progresspair_merge}:
            \begin{proof}
                $\ref{enum:progpairmerge}):$
                At the time, a progress pair $S,S'$ is started, the subboundary connecting both is a quasi line.
                Because of Lemma~\ref{lem:runinvariants}.$\ref{enum:runinvquasiline})$ and $\ref{enum:runinvforeignquasilines})$,
                this also does not change
                if other runs are located on this quasi line or next to it.
                And also merges preserve the quasi line properties.
                So, if not stopped, $S$ and $S'$ keep moving towards each other (Lemma~\ref{lem:runinvariants}.$\ref{enum:runinvmovement})$).
                Because of Lemma~\ref{lem:runinvariants}.$\ref{enum:runinvgoodpairs})$ a merge can be performed at the latest when they meet,
                which takes at most $n$ rounds.

                It remains to show that $S,S'$ are not stopped by the algorithm's termination conditions (Table~\ref{table:runterminate})
                before the merge could be performed.
                In the following, we check all these conditions.
                In this proof, we denote the quasi line, on which $S,S'$ are located by the variable $q$.
                Let $S^\star$ be the next sequent run in front of $S$.
                $\ref{enum:stoptooclose})$: By definition of a progress pair, no merge has been performed since the last time new runs have been started.
                This means, that if no run from outside the quasi line $q$ has interfered with $q$, the distance between $S$ and $S^\star$
                is at least $L-1$ (cf.\ the proof of Lemma~\ref{lem:runinvariants}.$\ref{enum:runinvops})$)
                which is bigger than the viewing radius. (The same analogously holds for $S'$.)
                So then, the runs of a progress pair cannot be stopped by this termination condition.
                So, we assume that other runs have interfered with $q$.
                If $S^\star$ originates from a different quasi line, $S$ and $S^\star$ can only be oriented in the way shown
                in Figure~\ref{fig:y-runs_contra}, because the run operations require the cells "`above'" be empty.
                If these quasi lines shall meet, then at least one of them cannot be $y$-monotone.
                But then a merge would have removed at lease one of them before both runs have come too close.
                So $S^\star$ cannot originate from a different quasi line.
                \begin{figure}[h]
                    \centering
                    \includegraphics[scale=\figscale]{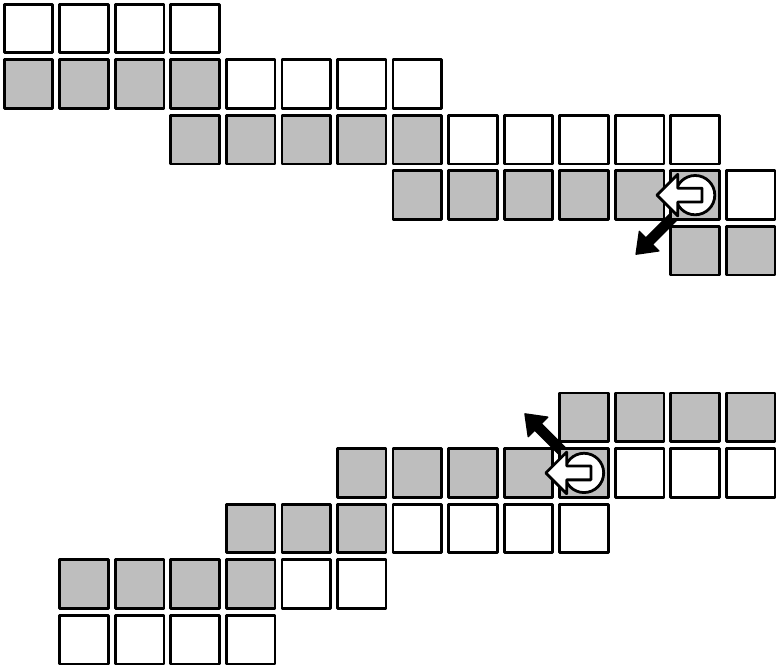}
                    \caption{Too close sequent runs cannot originate from different quasi lines.}
                    \label{fig:y-runs_contra}
                \end{figure}

                So, a run on a different quasi line (lying alongside $q$) must have shortened the part of $q$ that connects $S$ and $S^\star$.
                This can only have happened because of a merge operation, but this is not possible, because $S,S'$ is a progress pair.

                $\ref{enum:stopendpoint})$: If a run of a progress pair can see an endpoint of the quasi line in front of it, then the other run of the progress pair must have previously been stopped (Using the properties of Lemma~\ref{lem:runinvariants}.$\ref{enum:runinvquasiline})$ and $\ref{enum:runinvforeignquasilines})$.).
                Then, because a run of a progress pair cannot be stopped by condition $\ref{enum:stoptooclose})$, it instead must have been stopped because of a merge.
                So, a termination because of condition $\ref{enum:stopendpoint})$ is allowed.

                $\ref{enum:stoppass},\ref{enum:stoplongop})$:
                The change of the quasi lines shape could either have happened because of the reshapement of a sequent run or because of a merge operation.
                Because we have ensured the minimum distance of two sequent runs to be large enough, the first case cannot happened.
                So, a merge must have been the reason.
                Then, this merge must have been enabled by another run $S^\star$, moving towards $S$. 
                If $S^\star$ is the partner run of $S$, then stopping is allowed.
                Else, because progress pairs are nested into each other, $S^\star$ either was not a run of a progress pair or its partner run has previously been stopped.
                In the latter case, because runs of progress pairs cannot be stopped because of condition $\ref{enum:stoptooclose})$,
                the run must have been stopped by an earlier merge.
                In both cases, the current merge can be credited to the progress pair of $S$.

                $\ref{enum:stopmerge})$: By the same arguments, this merge does not need to be also credited to a different progress pair.
                This then also proves $\ref{enum:progpairuniqueness})$ of the lemma.
            \end{proof}
            The missing proof of Lemma~\ref{lem:runinvariants} can be found in Subsection~\ref{sec:runinvlemproof}.
        \subsection{Proof of Lemma~\ref{lem:runinvariants}}\label{sec:runinvlemproof}
            \begin{proof}
                Let $q$ be the quasi line $S$ is located on and we \oBdA\ assume that $q$ is a horizontal quasi line.
                $\ref{enum:runinvmovement})$: This directly follows by the run definition in Subsection~\ref{ssec:runs}.

                $\ref{enum:runinvquasiline})$: Cf.\ Figure~\ref{fig:ALG_hop_simple2}.
                OP-A ensures that on horizontal quasi lines only horizontal subchains of lengths $>2$ are shortened and vertical subchains remain unextended.
                OP-B does not reshape.
                OP-C can only executed during the first three rounds after $S$ has been started.
                Our design of the run passing operation also ensures that the quasi line properties are maintained.

                $\ref{enum:runinvforeignquasilines})$: 
                Because of the run stopping condition of Table~\ref{table:runterminate}.$\ref{enum:stopendpoint})$, $R(S)$ can only reshape
                a quasi line $q^\star$ if $q^\star$ is also a horizontal one.
                As the run operation OP-C cannot be applied inside $q^\star$, we only need to consider OP-A.
                If the cell onto that $R(S)$ hops, is empty, then OP-A (because of the cells that are required to be empty) can
                only be executed if $S$ is also located on $q^\star$.
                But then, by the same arguments as in $\ref{enum:runinvquasiline})$, the reshapements of $R(S)$ cannot violate the
                quasi line definition.
                If in contrast the cell onto the $R(S)$ hops is occupied by some of the robots of $q^\star$,
                then this operation does not modify the shape of $q^\star$.

                $\ref{enum:runinvvisseqrun})$: This is ensured by the run termination condition of Table~\ref{table:runterminate}.$\ref{enum:stoptooclose})$.

                $\ref{enum:runinvops})$: Cf.\ Figure~\ref{fig:ALG_hop_simple2}.
                All run operations OP-A, OP-B and OP-C ensure that if $S$ was located at some corner $c_1$ when the operation started,
                it afterwards either is located at some other corner $c_2$ such that both corners are rotated equally or terminates if the target
                corner has been removed during the operation (cf.~Table~\ref{table:runterminate}.$\ref{enum:stoplongop},\ref{enum:stoppass}))$.
                Then, because still located on a quasi line (cf.~$\ref{enum:runinvquasiline})$), again $a)$ or $b)$ can be applied or run passing is started.
                
                The run passing needs some closer look, because it interrupts other operations.
                In order to ensure a regulated behavior, we chose the distance between sequent runs big enough 
                such that a run does not have to execute a new run passing operation before it has finished its previous one.
                As the value of the constant $L$ controls the length of the waiting intervals between the start of two sequent runs
                and our stopping conditions maintain this minimum distance by stopping a run if the next sequent run in front of it becomes visible,
                we need to settle the values for these constants appropriately.
                At this place, we only analyze the case that was explained in Section~\ref{ssec:runs}, i.e., all participating runs are located
                on the same quasi line.
                The other cases (Section~\ref{sec:runpassingcomplete}) can be proven similarly.
                We look at two sequent runs $S_1$ and $S_1^{\mathrm{succ}}$ such that $S_1^{\mathrm{succ}}$ has been started after $S_1$.
                Their distance $D$ is at least $L-1$.
                This value is achieved if at the start of $S_1$ the operation OP-C (Figure~\ref{fig:ALG_hop_simple2}) was executed.
                
                Now, we chose the value for the constant $D$ big enough for ensuring that while some run $S_2$ is passing along $S_1$,
                $S_1^{\mathrm{succ}}$ becomes visible to $S_2$ the earliest when the passing operation with $S_1$ has been completed.
                Figure~\ref{fig:crossconst} shows an example for the longest possible duration of a run passing operation.
                \begin{figure}[h]
                    \centering
                    \includegraphics[scale=\figscale]{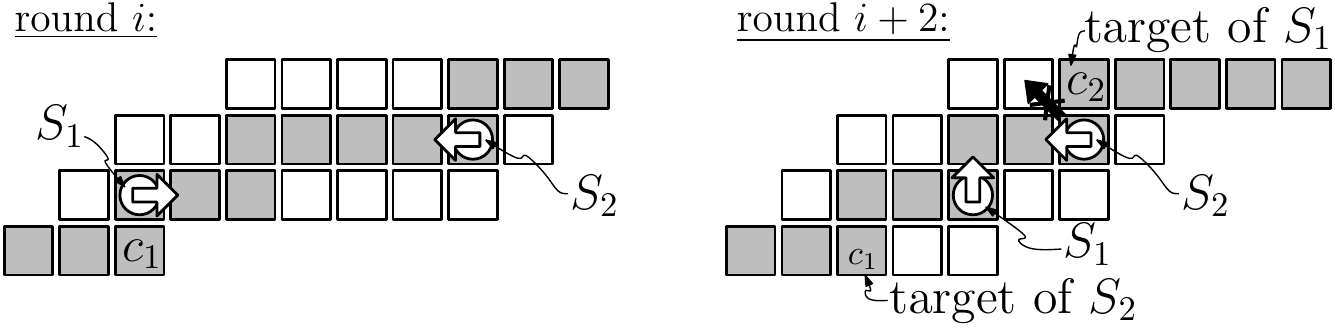}
                    \caption{Runs, passing along each other while the runner $R(S_1)$ is executing the operation OP-B (cf.\ Figure~\ref{fig:ALG_hop_simple2}).}
                    \label{fig:crossconst}
                \end{figure}

                Here, it takes $6$ rounds until $S_2$ has arrived at its target corner.
                Because the run passing operation starts when the distance between $S_1$ and $S_2$ is $\leq\crossdistance$, after the
                passing operation, the distance between $S_2$ and $S_1^{\mathrm{succ}}$ equals $D-9$.
                We want this then still be $\geq\crossdistance$.
                So we choose $D\geq 12$ and together with the above argumentation concerning the minimum distance between sequent runs
                follows $L\geq 13$.
                In order to detect that the distance has become smaller than $12$ (and solve this problem), the viewing radius must be $11$ (cf.\ Table~\ref{table:runterminate}.$\ref{enum:stoptooclose})$).

                Using the same construction, but including all other possible cases of run passings (Section~\ref{sec:runpassingcomplete}), we can show that (unoptimized) values of $L=\pipelininginterval$ a viewing radius of $\viewingradius$ suffice.

                $\ref{enum:runinvgoodpairs})$: When defining good pairs in Subsection~\ref{ssec:runs},
                good pairs have been characterized by the relative position of the outer direct neighbors of the good pair according to the quasi line.
                The first part of the proof of $\ref{enum:runinvops})$ finishes the proof.
            \end{proof}
    \section{Run Passing Operation in Detail}\label{sec:runpassingcomplete}
        If two runs $S,S'$ that are moving towards each other but do not enable a merge, have come too close to each other, then we let them pass along each other with a constant movement speed but
        without performing diagonal hops, in order to prevent the connectivity from breaking or destroying the quasi line property.
        We call this the \emph{run passing} operation.
        If they have come this close that the operation is needed, then also $S$ can see $S'$ withing its constant viewing range and vice versa.
        For the run passing operation, it is important to assign a target corner to every such run, so that it can continue its reshapements afterwards.
        \begin{figure}[h]
            \centering
            \includegraphics[scale=\figscale]{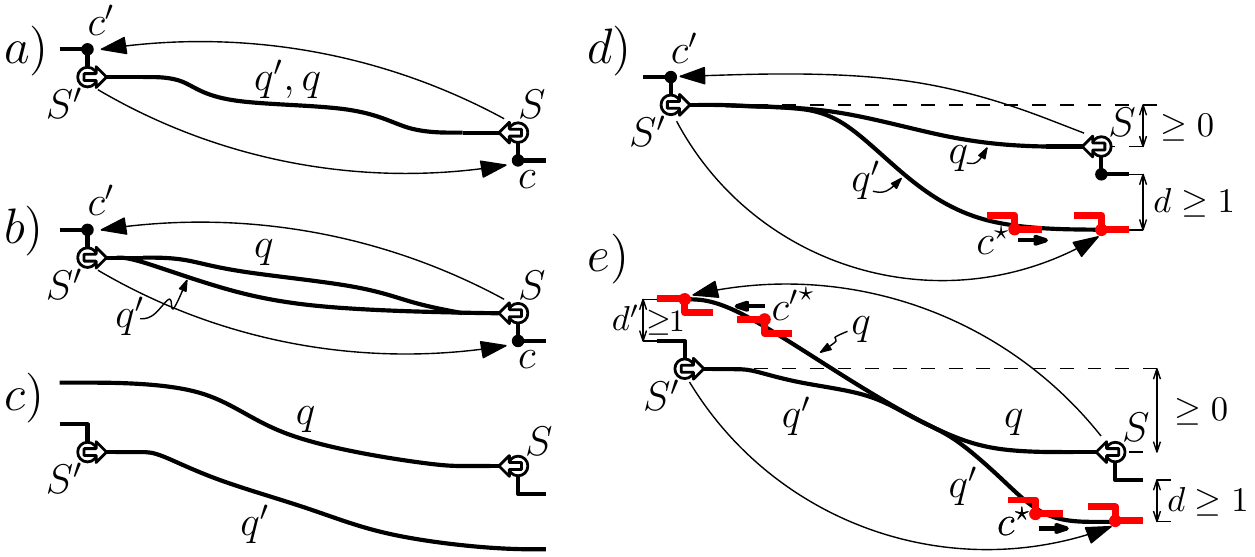}
            \caption{Classification of all possible cases that can occur if two runs need to pass along each other.}
            \label{fig:runcrossingfmadh}
        \end{figure}
        Cf.\ Figure~\ref{fig:runcrossingfmadh}.
        The subfigures show two runs $S,S'$ that need to pass along each other.
        The curve $q$, originating at $S$ symbolizes the quasi line that belongs to $S$. Analogously, the quasi line $q'$ belongs to $S'$.
        We can assume that locally no robots are located in the area above $q$ and below $q'$, because else the corresponding run $S$ respectively $S'$ could not be part of a good pair so that it is allowed to stop itself.
        Then, the run passing is not needed anymore.
        If a run passing operation is necessary, then, between the locations of $S,S'$, the quasi lines $q,q'$ must in some way overlap.
        We now explain that in all possible cases of overlappings always suitable target corners can be assigned to both of the runs, so that they can continue their reshapements after the passing.
        We refer to Figure~\ref{fig:runcrossingfmadh}.
        \begin{enumerate}[a)]
            \item $q,q'$ are identical. Then the run passing operation executes in the way explained in Section~\ref{ssec:runs}.\label{enum:qlident}
            \item $q,q'$ overlap at the location of $S$ and at the location of $S'$. Then by the same arguments as in \ref{enum:qlident}), the
            target corners $c,c'$ must exist.
            \item $q,q'$ are disjoint. Then the runners calculate whether their reshapements would make $q$ and $q'$ overlap.
            If so, then this overlapping would lead to the removal/merge of robots, so that we can credit this merge to the good pair and
            both runs are allowed to stop.
            Else, the runs continue their normal reshapement operations.
            \item $q,q'$ overlap only at the location of $S'$. Then, the vertical distance between $S$ and $S'$ must be $\geq0$ because else a merge would make a run passing superfluous.
            As on the right, $q'$ is located below $q$ ($d\geq1$), there must exist some corner $c^\star$ with the same rotation as the corner at that $S'$ is located.
            The robots then reconfigure the robots of $q'$ this way, that $c^\star$ is located at the right endpoint of the sub quasi line.
            Afterwards, $S'$ can move to this corner and then continue its reshapements.
            In contrast, $S$ can simply move to $c'$.\label{enum:qly}
            \item $q,q'$ overlap but are disjoint at both endpoints. Then, by the same arguments as in \ref{enum:qly}), both target corners exist.
        \end{enumerate}
\begin{sloppy}
\bibliography{references}

\newcommand{\etalchar}[1]{$^{#1}$}
\begin{thebibliography}{DDSKN12}


\providecommand{\url}[1]{\texttt{#1}}
\expandafter\ifx\csname urlstyle\endcsname\relax
  \providecommand{\doi}[1]{doi: #1}\else
  \providecommand{\doi}{doi: \begingroup \urlstyle{rm}\Url}\fi

\bibitem[ACLF{\etalchar{+}}15]{Jung2015b}
\textsc{Abshoff}, Sebastian ; \textsc{Cord-Landwehr}, Andreas ;
  \textsc{Fischer}, Matthias ; \textsc{Jung}, Daniel  ; \textsc{Heide},
  Friedhelm Meyer auf~d.:
\newblock Gathering a Closed Chain of Robots on a Grid.
\newblock {In: }\emph{CoRR} abs/1510.05454 (2015).
\newblock \url{http://arxiv.org/abs/1510.05454}

\bibitem[ACLF{\etalchar{+}}16]{Jung2016}
\textsc{Abshoff}, Sebastian ; \textsc{Cord-Landwehr}, Andreas ;
  \textsc{Fischer}, Matthias ; \textsc{Jung}, Daniel  ; \textsc{Heide},
  Friedhelm Meyer auf~d.:
\newblock Gathering a Closed Chain of Robots on a Grid.
\newblock {In: }\emph{IPDPS '16}, IEEE, to appear, Mai 2016

\bibitem[ASY95]{localgathering}
\textsc{Ando}, Hideki ; \textsc{Suzuki}, Yoshinobu  ; \textsc{Yamashita},
  Masafumi:
\newblock {Formation and agreement problems for synchronous mobile robots with
  limited visibility}.
\newblock {In: }\emph{ISIC '95}, 1995, S. 453--460

\bibitem[CFPS03]{gathering-icalp}
\textsc{Cieliebak}, Mark ; \textsc{Flocchini}, Paola ; \textsc{Prencipe},
  Giuseppe  ; \textsc{Santoro}, Nicola:
\newblock {Solving the Robots Gathering Problem}.
\newblock {In: }\emph{ICALP '03}, 2003, S. 1181--1196

\bibitem[CP04]{Cohen:2004a}
\textsc{Cohen}, Reuven ; \textsc{Peleg}, David:
\newblock {Robot Convergence via Center-of-Gravity Algorithms}.
\newblock {In: }\emph{SIROCCO '04} Bd. 3104, 2004 ({LNCS}), S. 79--88

\bibitem[DDSKN12]{gatheringongrids}
\textsc{D’Angelo}, Gianlorenzo ; \textsc{Di~Stefano}, Gabriele ;
  \textsc{Klasing}, Ralf  ; \textsc{Navarra}, Alfredo:
\newblock Gathering of Robots on Anonymous Grids without Multiplicity
  Detection.
\newblock {In: }\emph{SIROCCO '12} Bd. 7355.
\newblock 2012. --
\newblock ISBN 978--3--642--31103--1, S. 327--338

\bibitem[DFKP06]{rendezvousingraphen}
\textsc{Dessmark}, Anders ; \textsc{Fraigniaud}, Pierre ; \textsc{Kowalski},
  Dariusz~R.  ; \textsc{Pelc}, Andrzej:
\newblock {Deterministic Rendezvous in Graphs}.
\newblock {In: }\emph{Algorithmica} 46 (2006), Nr. 1, S. 69--96

\bibitem[DKL{\etalchar{+}}11]{gatheringthetanquadrat}
\textsc{Degener}, Bastian ; \textsc{Kempkes}, Barbara ; \textsc{Langner},
  Tobias ; \textsc{{Meyer auf der Heide}}, Friedhelm ; \textsc{Pietrzyk}, Peter
   ; \textsc{Wattenhofer}, Roger:
\newblock A tight runtime bound for synchronous gathering of autonomous robots
  with limited visibility.
\newblock {In: }\emph{SPAA '11}, 2011, S. 139--148

\bibitem[DKLH06]{gtm}
\textsc{Dynia}, Miroslaw ; \textsc{Kutylowski}, Jaroslaw ; \textsc{Lorek},
  Pawel  ; \textsc{Heide}, Friedhelm Meyer auf~d.:
\newblock Maintaining Communication Between an Explorer and a Base Station.
\newblock {In: }\emph{IFIP TC10}, 2006, S. 137--146

\bibitem[DKM10]{MINCH}
\textsc{Degener}, Bastian ; \textsc{Kempkes}, Barbara  ; \textsc{{Meyer auf der
  Heide}}, Friedhelm:
\newblock {A local $O(n^2)$ gathering algorithm}.
\newblock {In: }\emph{SPAA '10}, 2010, S. 217--223

\bibitem[FPS12]{flocchinioverview}
\textsc{Flocchini}, Paola ; \textsc{Prencipe}, Giuseppe  ; \textsc{Santoro},
  Nicola:
\newblock \emph{Distributed Computing by Oblivious Mobile Robots}.
\newblock Morgan {\&} Claypool, 2012 (Synthesis Lectures on Distributed
  Computing Theory)

\bibitem[ISK{\etalchar{+}}12]{Izumi2012}
\textsc{Izumi}, Taisuke ; \textsc{Souissi}, Samia ; \textsc{Katayama}, Yoshiaki
  ; \textsc{Inuzuka}, Nobuhiro ; \textsc{D\'{e}fago}, Xavier ; \textsc{Wada},
  Koichi  ; \textsc{Yamashita}, Masafumi:
\newblock {The Gathering Problem for Two Oblivious Robots with Unreliable
  Compasses}.
\newblock {In: }\emph{SICOMP} 41 (2012), Nr. 1, S. 26--46

\bibitem[KM09]{hopper}
\textsc{Kutylowski}, Jaroslaw ; \textsc{{Meyer auf der Heide}}, Friedhelm:
\newblock {Optimal strategies for maintaining a chain of relays between an
  explorer and a base camp}.
\newblock {In: }\emph{TCS} 410 (2009), Nr. 36, S. 3391--3405

\bibitem[KMP08]{gatheringOnRing}
\textsc{Klasing}, Ralf ; \textsc{Markou}, Euripides  ; \textsc{Pelc}, Andrzej:
\newblock {Gathering asynchronous oblivious mobile robots in a ring}.
\newblock {In: }\emph{TCS} 390 (2008), Nr. 1, S. 27--39

\bibitem[KTI{\etalchar{+}}07]{gathering-compasses}
\textsc{Katayama}, Y ; \textsc{Tomida}, Y ; \textsc{Imazu}, H ;
  \textsc{Inuzuka}, N  ; \textsc{Wada}, Koichi:
\newblock {Dynamic Compass Models and Gathering Algorithms for Autonomous
  Mobile Robots}.
\newblock {In: }\emph{SIROCCO '07} Bd. 4474, 2007 ({LNCS}), S. 274--288

\bibitem[Mar09]{practicalrendevouzaktuell}
\textsc{Mart\^{\i}nez}, Sonia:
\newblock {Practical multiagent rendezvous through modified circumcenter
  algorithms}.
\newblock {In: }\emph{Automatica} 45 (2009), Nr. 9, S. 2010--2017

\bibitem[Pre07]{impossibilityofgathering}
\textsc{Prencipe}, Giuseppe:
\newblock {Impossibility of gathering by a set of autonomous mobile robots}.
\newblock {In: }\emph{TCS} 384 (2007), Nr. 2-3, S. 222--231

\bibitem[SN13]{Stefano2013}
\textsc{Stefano}, Gabriele~D. ; \textsc{Navarra}, Alfredo:
\newblock {Optimal Gathering of Oblivious Robots in Anonymous Graphs}.
\newblock {In: }\emph{LNCS} 8179 (2013), S. 213--224

\bibitem[SN14]{OptExactGatheringGrids2014}
\textsc{Stefano}, Gabriele~D. ; \textsc{Navarra}, Alfredo:
\newblock {Optimal Gathering on Infinite Grids}.
\newblock {In: }\emph{SSS '14}.
\newblock 2014, S. 211--225

\end{thebibliography}
\end{sloppy}
\end{document}